\documentclass[draftclsnofoot,onecolumn,11pt]{IEEEtran}

\usepackage{amsmath,amssymb,amsfonts}
\usepackage{array}
\usepackage[caption=false,font=normalsize,labelfont=sf,textfont=sf]{subfig}
\usepackage{textcomp}
\usepackage{stfloats}
\usepackage{verbatim}
\usepackage{graphicx}
\usepackage{cite}
\usepackage{tabularx}
\usepackage{booktabs}

\usepackage{acronym}
\acrodef{5G}{the fifth generation}
\acrodef{MIMO}{multiple-input multiple-output}
\acrodef{MISO}{multiple-input single-output}
\acrodef{MU}{multi-user}
\acrodef{EE}{energy efficiency}
\acrodef{SE}{spectral efficiency}
\acrodef{RF}{radio-frequency}
\acrodef{BS}{base station}
\acrodef{UE}{user equipment}
\acrodef{AQNM}{additive quantization noise model}
\acrodef{DAC}{digital-to-analog converter}
\acrodef{PA}{power amplifier}
\acrodef{SINR}{signal-to-interference-plus-noise ratio}
\acrodef{SNR}{signal-to-noise ratio}
\acrodef{WMMSE}{weighted minimum mean square error}
\acrodef{LoS}{line-of-sight}
\acrodef{NLoS}{non-line-of-sight}
\acrodef{AoA}{angle-of-arrival}
\acrodef{AoD}{angle-of-departure}
\acrodef{UPA}{uniform planar array}
\acrodef{ARV}{array response vector}
\acrodef{CGV}{channel gain vector}
\acrodef{AGV}{antenna gain vector}
\acrodef{EM}{electromagnetic}
\acrodef{CSI}{channel state information}
\acrodef{OFDM}{orthogonal frequency-division multiplexing}
\acrodef{RIS}{reconfigurable intelligent surface}
\acrodef{MA}{movable antenna}
\acrodef{3D}{three-dimensional}
\acrodef{AWGN}{additive white Gaussian noise}
\acrodef{ERA}{electromagnetically reconfigurable antenna}
\acrodef{MiLAC}[MiLAC]{microwave linear analog computer}
\acrodef{LMI}[LMI]{linear matrix inequality}
\acrodef{SDP}[SDP]{semi-definite programming}
\acrodef{PGD}[PGD]{projected gradient descent}
\acrodef{SVD}[SVD]{singular value decomposition}
\acrodef{MOO}[MOO]{multi-objective optimization}
\acrodef{KKT}[KKT]{Karush-Kuhn-Tucker}
\acrodef{FP}[FP]{fractional programming}
\acrodef{PS}[PS]{phase shifter}
\acrodef{SC}[SC]{sub-connected}
\acrodef{FC}[FC]{fully-connected}
\acrodef{CRB}[CRB]{Cram\'er--Rao bound}
\acrodef{DoA}[DoA]{direction-of-arrival}
\acrodef{ULA}[ULA]{uniform linear array}
\acrodef{ADC}[ADC]{analog-to-digital converter}
\acrodef{DFT}[DFT]{discrete Fourier transform}
\acrodef{PD}[PD]{positive definite}
\acrodef{PSD}[PSD]{positive semi-definite}
\acrodef{ISAC}[ISAC]{integrated sensing and communication}
\acrodef{FIM}[FIM]{Fisher information matrix}
\acrodef{MLE}[MLE]{maximum-likelihood estimator}
\acrodef{MSE}[MSE]{mean square error}
\usepackage{amsthm}
\newtheoremstyle{zackplain}%
  {\topsep}{\topsep}{\itshape}{}{\bfseries}{:}{.5em}{}
\theoremstyle{zackplain}
\newtheorem{lemma}{\textbf{Lemma}}
\newtheorem{theorem}{\textbf{Theorem}}
\newtheorem{proposition}{\textbf{Proposition}}
\newtheorem{definition}{\textbf{Definition}}
\newtheorem{remark}{\textbf{Remark}}

\newtheorem{corollary}{\textbf{Corollary}}

\usepackage{tikz} 
\usetikzlibrary{calc, positioning, fit, backgrounds, decorations.pathmorphing, arrows.meta, shapes.geometric, shapes.misc, patterns}
\usepackage[utf8]{inputenc}
\usepackage{pgfplots}
\pgfplotsset{compat=1.18}
\usepgfplotslibrary{groupplots}
\pgfplotsset{
    /pgfplots/ieee axis/.style={
        width=\linewidth,
        height=0.78\linewidth,
        tick align=outside,
        tick pos=left,
        enlarge x limits=false,
        axis line style={line width=0.5pt},
        tick style={line width=0.4pt, black},
        major tick length=2.4pt,
        minor tick length=1.2pt,
        grid=both,
        major grid style={line width=0.25pt, draw=black!20},
        minor grid style={line width=0.12pt, draw=black!10},
        label style={font=\footnotesize},
        tick label style={font=\scriptsize},
        title style={font=\footnotesize},
        legend style={
            font=\scriptsize,
            draw=black!50,
            line width=0.3pt,
            fill=white,
            fill opacity=0.85,
            text opacity=1,
            inner sep=1.8pt,
            row sep=-2pt,
        },
        unbounded coords=jump,
    },
}
\pgfplotsset{
    every axis label/.append style={font=\footnotesize},
}
\usepackage{tabu,longtable}
\usepackage{diagbox}
\usepackage{threeparttable}
\usepackage{makecell}
\usepackage{multirow}
\usepackage[letterpaper,margin=1in]{geometry}
\usepackage{units}
\usepackage{bm}
\usepackage{amssymb}
\newcommand{\TT}{\mathsf{T}}
\newcommand{\HH}{\mathsf{H}}

\newcommand{\av}{{\bf a}}

\newcommand{\cv}{{\bf c}}

\newcommand{\gv}{{\bf g}}

\newcommand{\nv}{{\bf n}}

\newcommand{\sv}{{\bf s}}

\newcommand{\vv}{{\bf v}}
\newcommand{\xv}{{\bf x}}
\newcommand{\yv}{{\bf y}}
\newcommand{\zv}{{\bf z}}

\newcommand{\Am}{{\bf A}}
\newcommand{\Bm}{{\bf B}}
\newcommand{\Cm}{{\bf C}}
\newcommand{\Dm}{{\bf D}}
\newcommand{\Em}{{\bf E}}
\newcommand{\Fm}{{\bf F}}
\newcommand{\Gm}{{\bf G}}

\newcommand{\Id}{{\bf I}}
\newcommand{\Jm}{{\bf J}}

\newcommand{\Mm}{{\bf M}}
\newcommand{\Nm}{{\bf N}}

\newcommand{\Pm}{{\bf P}}
\newcommand{\Qm}{{\bf Q}}
\newcommand{\Rm}{{\bf R}}

\newcommand{\Um}{{\bf U}}
\newcommand{\Wm}{{\bf W}}
\newcommand{\Vm}{{\bf V}}
\newcommand{\Xm}{{\bf X}}
\newcommand{\Ym}{{\bf Y}}

\newcommand{\Lambdam}{\hbox{\boldmath$\Lambda$}}

\newcommand{\Sigmam}{\hbox{\boldmath$\Sigma$}}

\newcommand{\Thetam}{\hbox{\boldmath$\Theta$}}

\providecommand{\C}{\mathbb{C}}
\providecommand{\R}{\mathbb{R}}

\providecommand{\row}{\operatorname{row}}
\providecommand{\rank}{\operatorname{rank}}

\providecommand{\diag}{\operatorname{diag}}
\providecommand{\vect}[1]{\boldsymbol{#1}}
\providecommand{\MiLAC}{\mathrm{MiLAC}}
\providecommand{\dig}{\mathrm{dig}}

\providecommand{\crb}{\mathrm{CRB}}
\providecommand{\col}{\operatorname{col}}
\providecommand{\mat}[1]{\mathbf{#1}}

\usepackage{amsmath}
\DeclareMathOperator*{\argmax}{arg\,max}

\usepackage[section]{placeins}

\pgfplotsset{
    /pgfplots/ieee axis/.append style={
        width=12cm,
        height=7.5cm,
        tick label style={font=\scriptsize},
        label style={font=\footnotesize},
        title style={font=\footnotesize},
    },
}

\begin{document}

\title{How Many RF Chains Does a Microwave Linear Analog Computer (MiLAC) Need to Match the Fully-Digital Cram\'er-Rao Bound?}

\author{Yuchen~Zhang, \emph{Member, IEEE},  Yu~Ge, \emph{Member, IEEE}, \\Bruno Clerckx, \emph{Fellow, IEEE}, and Tareq~Y.~Al-Naffouri, \emph{Fellow, IEEE}
\thanks{Y. Zhang and T. Y. Al-Naffouri are with the Computer, Electrical and Mathematical Sciences and Engineering Division, King Abdullah University of Science and Technology (KAUST), Thuwal 23955-6900, Saudi Arabia (e-mail: \{yuchen.zhang, tareq.alnaffouri\}@kaust.edu.sa).}
\thanks{Y. Ge is with Massachusetts Institute of Technology (MIT), Cambridge, MA 02139 USA (e-mail: yuge@mit.edu)}%
\thanks{B. Clerckx is with the Department of Electrical and Electronic Engineering, Imperial College London, SW7 2AZ London, U.K. (e-mail: b.clerckx@imperial.ac.uk).}
}

\maketitle

\begin{abstract}
A microwave linear analog computer (MiLAC) is a tunable microwave network that exploits wave propagation to perform computation directly in the analog domain at microwave frequencies. A recent application is to use MiLAC as the analog front end of an antenna array, where the antenna-to-radio-frequency (RF) chain mapping constitutes a linear operation that can be physically realized in the analog domain. 
The active chain count then scales with the number of data streams rather than with the number of antennas. Prior work establishes the lossless reciprocal MiLAC, which avoids power dissipation and non-reciprocal components, as a beamforming-flexible (or even capacity-achieving) front end for wireless communication, but its sensing performance has remained largely unexplored. This paper delivers one of the earliest Cram\'er--Rao bound (CRB)-based analyses of direction-of-arrival estimation under a tunable, receive-side lossless reciprocal MiLAC combiner for $K$ far-field targets. We show that the Fisher information matrix depends on the analog combiner only through the orthogonal projector onto its row space, descending the optimization from the matrix manifold to the Grassmannian. The MiLAC Fisher information never exceeds that of a fully-digital receiver, with equality whenever the combiner's row space contains a $2K$-dimensional joint steering--derivative subspace. This yields a zero-gap threshold of two RF chains per target. On the hardware side, a dimension-counting argument lower-bounds the tunable-component count of a MiLAC class that achieves the digital CRB for every target configuration. The stem-connected MiLAC architecture, whose number of tunable components scales linearly in both the antenna and target counts, attains this bound asymptotically and up to an antenna-count-independent additive overhead. MiLAC attains the fully-digital CRB exactly, whereas a phase-shifter front end of the same RF chain count generally leaves a gap. Numerical experiments confirm every claim.
\end{abstract}


\section{Introduction}
\label{sec:intro}

\IEEEPARstart{T}{he} sixth-generation (6G) wireless network is envisioned to deliver not only high-speed communication but also high-resolution environmental sensing as a native, first-class capability. \Ac{ISAC}, in which the same hardware, spectrum, and waveforms simultaneously carry information bits and probe the physical environment, has consequently emerged as one of the defining technological pillars of 6G~\cite{Liu2022Survey}. Beyond providing a ``free'' sensing capability on top of the existing communication infrastructure, \ac{ISAC} is expected to support emerging applications such as autonomous driving, low-altitude unmanned aerial vehicles, industrial automation, smart factories, and human-centric sensing, all of which demand high spatial resolution and robust target discrimination in dense multi-target scenes. This vision has fueled an intense research effort spanning the information-theoretic tradeoff between communication and sensing~\cite{Xiong2023Fundamental}, joint beamforming with explicit \ac{CRB} constraints~\cite{Liu2022CRBOpt}, and propagation-engineered sensing aided by \ac{RIS}~\cite{Zheyu2025tit} or metamaterials~\cite{Demirhan2023ISAC}.

A recurring conclusion of this line of work is that the sharpest improvements in sensing accuracy come from \emph{scale}, namely, larger apertures, more antennas, and therefore more spatial degrees of freedom~\cite{Bjornson2025Gigantic,Larsson2014Massive}. Receive arrays with antenna counts in the hundreds to the thousands are thus expected to become the natural backbone of both communication and sensing in the upper mid-band~\cite{Bjornson2025Gigantic}. A fully-digital implementation, in which every antenna element is followed by a dedicated low-noise amplifier, down-converter, and high-resolution \ac{ADC}, becomes prohibitive at this scale, since the per-antenna hardware, power, and data-rate cost all grow linearly with the antenna count~\cite{Heath2016Overview}. Closing the gap between such large-array deployments and an implementable \ac{RF} front end is one of the central hardware-efficiency problems of 6G \ac{ISAC}.

\subsection{From phase shifters to MiLAC}

The textbook response to this bottleneck has been \emph{hybrid} analog-digital beamforming, in which a phase-shifter network compresses the antenna signal down to a much smaller number of active \ac{RF} chains~\cite{ElAyach2014Spatially}. Each entry of the analog factor is constrained to constant modulus, which restricts the reachable analog mappings to a strict, lower-dimensional subset of those an unconstrained analog combiner could realize. More recent reconfigurable-hardware proposals, such as dynamic/holographic metasurfaces~\cite{Shlezinger2021Dynamic,Ruizhi2026ICASSP,Zhuoyang2026TSP}, pattern-reconfigurable antennas~\cite{Alireza2026JSTSP,Wenyan2026ST}, and tri-hybrid architectures~\cite{Castellanos2026Embracing,Pinjun2025Tri,Jiangong2026Tri}, add an extra reconfigurable layer, such as reconfigurable antenna elements, on top of the phase-shifter network and thereby widen its design flexibility. Because the analog combining stage itself is still phase-shifter-based, however, this constant-modulus restriction persists.

A \ac{MiLAC} is a radically different front-end paradigm~\cite{Nerini2025AnalogI,Nerini2025AnalogII,Zack2026EE,Wu2026MiLAC,Nerini2025Reduced}. In its general form it is a tunable multiport microwave network whose internal admittance components, when properly set, implement a prescribed linear transformation on \ac{RF} signals through the scattering behavior of the network itself~\cite{Nerini2025AnalogI}. For a receive array, the output ports of the \ac{MiLAC} drive a small number of active \ac{RF} chains while the input ports terminate the antenna elements, so that the antenna-to-chain mapping is executed through the passive network rather than in digital or active analog hardware~\cite{Nerini2025AnalogII}. Communication-oriented designs further restrict the tunable admittances to be purely imaginary (\emph{lossless}, so that no signal energy is dissipated) and symmetric (\emph{reciprocal}, so that only ordinary two-terminal reactive components, rather than circulators or isolators, are needed)~\cite{Nerini2025Capacity,Wu2026MiLAC}, under which the multiport scattering matrix becomes symmetric unitary. 

Two features distinguish this architecture from digital and phase-shifter-based counterparts. First, because only the \ac{RF} chains after the \ac{MiLAC} network are active, the amount of active hardware it requires is set by the number of \ac{RF} chains rather than the number of antennas. By contrast, a fully digital transceiver requires active hardware whose count grows linearly with the antenna array size. Second, the antenna-to-chain mapping is implemented by a fully reactive network. As a result, any analog beamformer or combiner with spectral norm no larger than one is, in principle, realizable through an appropriate tuning of the internal susceptances~\cite{Nerini2025Capacity,Wu2026MiLAC}. This feasible set is strictly richer than the constant-modulus row-constrained set imposed by phase-shifter-based beamformers/combiners with the same number of \ac{RF} chains. The advantage appears both in the unconstrained matrix space and, more importantly, in the row spaces that can be realized. This additional flexibility underpins the capacity-achieving and beamforming-flexible designs reported in~\cite{Nerini2025Capacity,Wu2026MiLAC,Nerini2025Reduced,Zack2026Stem}.

Research on the \ac{MiLAC} paradigm has developed rapidly along two complementary threads. The first treats \ac{MiLAC} as a general-purpose linear analog computer, and the foundational theory in~\cite{Nerini2025AnalogI,Nerini2025AnalogII} showed that, with unconstrained admittances, a tunable multiport network can implement general linear operations at the speed of light. The second applies \ac{MiLAC} to beamforming under the lossless and reciprocal constraints. In the point-to-point setting, fully-connected \ac{MiLAC} architectures achieve full \ac{MIMO} capacity~\cite{Nerini2025Capacity}, and reduced-complexity \emph{stem-connected} topologies lower the admittance count from quadratic to linear in the antenna count while preserving capacity~\cite{Nerini2025Reduced}. In the multiuser multiple-input single-output setting, the feasible beamformer set has been characterized~\cite{Wu2026MiLAC}, a performance-limit analysis has quantified the gap to fully-digital beamforming as the array grows~\cite{Fang2026Performance}, and complementary efforts address analog-domain channel estimation~\cite{Zhang2026Channel}, physics-compliant modeling with mutual coupling~\cite{Nerini2026Physics}, wideband \ac{OFDM} beamforming~\cite{Peng2026OFDM}, two-layer transmit architectures for multiuser networks~\cite{Zhou2026TwoLayer}, simultaneous active and passive (\ac{RIS}-like) beamforming~\cite{Nerini2026ActivePassive}, lossy designs that account for the dissipation of practical tunable components~\cite{Zhou2026Lossy}, hardware realizations of analog computing through hybrid couplers and phase shifters~\cite{Nerini2026HybridCouplers}, and an initial study of \ac{MiLAC}-aided transmit beamforming and \ac{DFT}-based receive processing for \ac{MIMO} radar sensing~\cite{Liu2026Radar}.

\subsection{An open question: Is MiLAC also a good \emph{sensing} front end?}
\label{subsec:open-question}

Despite this rapid progress, \ac{MiLAC} studies have so far predominantly optimized communication-theoretic objectives such as rate, capacity, or beamforming flexibility. A recent exception~\cite{Liu2026Radar} considers \ac{MiLAC}-aided \ac{MIMO} radar sensing on the transmit side and a fixed receiver-side two-dimensional \ac{DFT} implemented within the \ac{MiLAC}, but the Fisher information that a tunable receive-side \ac{MiLAC} combiner preserves about the physical environment remains uncharacterized. The following question is therefore open:

\begin{quote}\it
If a receive array is implemented through a lossless reciprocal MiLAC with only a small number of RF chains, how much Fisher information about the angles of $K$ impinging targets is lost, and what is the minimum hardware complexity at which this loss can be driven to zero?
\end{quote}

The answer is not a corollary of the capacity-preservation results above. Capacity arguments concern the information content about the transmitted data symbols, whereas sensing is about the information content of the observation about the physical parameters of the environment, such as angles, delays, Doppler shifts, and complex amplitudes. Even when the data rate is preserved, a \ac{MiLAC} front end generally discards Fisher information about these physical parameters whenever the orthogonal projection onto its row space acts non-trivially on the relevant signal subspace. Answering the question above is therefore a prerequisite for deploying \ac{MiLAC}-aided receivers in \ac{ISAC} and radar applications with large receive arrays~\cite{Li2007MIMORadar,Liu2022CRBOpt,Xiong2023Fundamental}, for which the \ac{DoA} of a small number of far-field sources is the canonical benchmark estimation target~\cite{Stoica2005,Schmidt1986MUSIC}.

This paper answers the question above through a chain of results that connects the information-theoretic \ac{CRB} to physically realizable reduced-complexity hardware. To make the story concrete, we focus on $K$-target \ac{DoA} estimation with a \ac{ULA}. The main message is encouraging. Two active \ac{RF} chains \emph{per target}, with properly oriented \ac{MiLAC} combiner, already suffice to match the \ac{CRB} of a fully-digital receiver, and this optimum can be synthesized by a stem-connected \ac{MiLAC} topology whose hardware complexity scales linearly with both the antenna count and the target count. 

\subsection{Contributions}

In more detail, this paper makes the following contributions.

\begin{itemize}
\item \textbf{First Fisher-information characterization of tunable \ac{MiLAC} receive combining:}
We open this line of inquiry for $K$-target \ac{DoA} estimation, and the answer is favorable: as few as two \ac{RF} chains per target retain the \emph{entire} Fisher information of a fully-digital receiver, with a hardware realization whose tunable-component count grows linearly in both the antenna and target counts.\footnote{The intuition behind the factor of two is that the observation changes with each target's parameters along only two directions: varying the complex amplitude moves the array response along the steering vector, while varying the angle moves it along the derivative of that steering vector (made precise in Section~\ref{subsec:toy}). A combiner that retains both directions for all $K$ targets, i.e., all $2K$, discards no information about the angles, whereas one with fewer chains is forced to drop at least one of them. The second direction per target is the price of the \emph{unknown} amplitudes: were the amplitudes known, the $K$ derivative directions alone would suffice. A constant-modulus phase-shifter combiner with the very same $2K$ chains cannot, in general, align its row space to these target-dependent directions, and therefore falls short of the fully-digital bound for almost every target geometry, as will be detailed in Section~\ref{subsec:ps-strict}.}

\item \textbf{Digital-\ac{CRB} achievability under \ac{MiLAC} combining:}
We characterize when a lossless reciprocal \ac{MiLAC} preserves the fully-digital Fisher information for $K\ge 1$ targets. A compact \ac{FIM} expression depends on the combiner only through its row-space projector (Theorem~\ref{thm:fim-projector} and Proposition~\ref{prop:subspace}), descending the optimization from the matrix manifold to the complex Grassmannian. The \ac{MiLAC} \ac{FIM} is no larger than its fully-digital counterpart in the L\"owner sense, with equality if and only if the combiner's row space contains the joint steering--derivative subspace of dimension $2K$ (Theorem~\ref{thm:lowner}), yielding a zero-gap threshold of two \ac{RF} chains per target (Corollary~\ref{cor:zerogap}) cleanly separated from the identifiability threshold of $\lceil 3K/2\rceil$ chains (Remark~\ref{rem:thresholds}).

\item \textbf{Reduced-hardware-complexity \ac{MiLAC} realizations:}
A symmetric square-root construction embeds every row-isometric combiner as the off-diagonal block of a symmetric unitary scattering matrix (Lemma~\ref{lem:reachability}), so every optimal row space is realizable by a lossless reciprocal \ac{MiLAC}. A dimension-counting argument lower-bounds the tunable-component count of any \ac{MiLAC} class that achieves the digital \ac{CRB} for every target configuration (Proposition~\ref{prop:lower-bound}), which the stem-connected architecture of~\cite{Nerini2025Reduced} attains up to an antenna-count-independent overhead (Theorem~\ref{thm:stem}), via a closed-form synthesis linear in the antenna and target counts. Thus the stem-connected \ac{MiLAC}, previously shown to achieve the fundamental limits of communication~\cite{Nerini2025Reduced}, is shown here to be beneficial for sensing as well. A complementary argument shows that, at the same \ac{RF}-chain budget, a fixed phase-shifter combiner leaves a strictly positive gap to the \ac{MiLAC} for almost every target configuration (Proposition~\ref{prop:ps-strict}).
\end{itemize}

Monte Carlo experiments on a half-wavelength \ac{ULA} validate every theoretical claim.

\emph{Paper organization and Notations:}
Section~\ref{sec:system_model} introduces the $K$-target signal model and the lossless reciprocal \ac{MiLAC} feasibility set. Section~\ref{sec:crb-general} develops the \ac{CRB} theory: the \ac{FIM} expression, its row-space invariance, the L\"owner ordering and zero-gap theorem, the row-isometry reachability lemma, the \ac{ULA} aperture scaling, and the strict inferiority of phase-shifter combiners. Section~\ref{sec:hardware} formulates and resolves the complexity-reduction question through the Stiefel-universal lower bound and the stem-connected attainment result. Section~\ref{sec:numerical} reports numerical validation, Section~\ref{sec:conclusion} concludes, and the Appendix collects all proofs.

Scalars are denoted by lowercase letters, vectors by bold lowercase letters, and matrices by bold uppercase letters. The Euclidean norm is $\|\cdot\|$ and the spectral norm is $\|\cdot\|_2$. Transpose, complex conjugate, and Hermitian transpose are $(\cdot)^{\TT}$, $\bar{(\cdot)}$, and $(\cdot)^{\HH}$. For a matrix $\Am$, $\rank(\Am)$ is its rank and $\diag(\av)$ denotes the diagonal matrix with entries $\av$. The row space, column space, and kernel of $\Am$ are $\row(\Am)$, $\col(\Am)$, and $\ker(\Am)$, where we adopt the inner-product convention $\row(\Am):=\col(\Am^{\HH})$ so that the Moore--Penrose projector $\Am^{\HH}(\Am\Am^{\HH})^{-1}\Am$ projects onto $\row(\Am)$, which coincides with the literal span of $\Am$'s rows when $\Am$ is real. The L\"owner partial order on Hermitian matrices is written $\Am\preceq\Bm$ iff $\Bm-\Am$ is positive semidefinite. The real and imaginary parts of a scalar are $\Re\{a\}$ and $\Im\{a\}$, and the Kronecker product is $\otimes$. The $n\times n$ identity matrix is $\Id_n$, and $\mathcal{CN}(\av,\Cm)$ denotes a circularly symmetric complex Gaussian distribution with mean $\av$ and covariance $\Cm$. The complex Stiefel manifold of $L_R\times N_R$ row-isometric matrices is $\mathrm{St}(L_R,N_R):=\{\Gm\in\C^{L_R\times N_R}:\Gm\Gm^{\HH}=\Id_{L_R}\}$, and the Grassmannian of $L_R$-dimensional subspaces of $\C^{N_R}$ is $\mathrm{Gr}(L_R,N_R)$. The orthogonal projector of $\C^{N_R}$ onto a subspace $\mathcal{S}$ is $\Pm_{\mathcal{S}}$.

\section{System Model}
\label{sec:system_model}

This section sets up the $K$-target signal model, the lossless reciprocal \ac{MiLAC} feasibility set, and the fully-digital baseline \ac{CRB} used as the benchmark in the sequel.

\subsection{Signal model}
\label{subsec:signal-model}

As shown in Fig.~\ref{fig:system}, we consider $K$ far-field targets at distinct azimuths $\theta_1,\ldots,\theta_K\in\R$ with unknown complex amplitudes $\beta_1,\ldots,\beta_K\in\C$ ($\beta_k\neq 0$ for all $k$), reflecting a common known waveform $\{s_t\}_{t=1}^{T}$ onto an $N_R$-element receive array with steering vector $\av(\theta)\in\C^{N_R}$. Stack the angles and amplitudes as $\vect{\theta}:=[\theta_1,\ldots,\theta_K]^{\TT}\in\R^{K}$ and $\vect{\beta}:=[\beta_1,\ldots,\beta_K]^{\TT}\in\C^{K}$, and define $\Am(\vect{\theta}):=[\av(\theta_1),\ldots,\av(\theta_K)]\in\C^{N_R\times K}$. The antenna-domain snapshot at time $t$ writes as
\begin{equation}
\label{eq:antenna-obs}
\yv_t=\Am(\vect{\theta})\vect{\beta}\,s_t+\nv_t,\qquad t=1,\dots,T,
\end{equation}
where $\nv_{t}\sim\mathcal{CN}(\mathbf{0},\sigma^{2}\Id_{N_R})$ is the pre-combining additive white Gaussian noise with power $\sigma^{2}$.

An $(N_R{+}L_R)$-port lossless reciprocal \ac{MiLAC} front end, whose feasible set is specified in Section~\ref{subsec:whatismilac}, connects $N_R$ ports to the antennas and $L_R$ ports to the \ac{RF} chains, with $L_R\le N_R$. It applies a deterministic linear combiner $\Gm\in\C^{L_R\times N_R}$, yielding the observation
\begin{equation}
\label{eq:obs}
\zv_{t}\,=\,\Gm\Am(\vect{\theta})\vect{\beta}\,s_{t}+\Gm\nv_{t},\qquad t=1,\dots,T.
\end{equation}
The post-combining noise has covariance $\Rm_{\Gm}:=\Gm\Gm^{\HH}$, which is in general not a scaled identity. In particular, $\Gm\Gm^{\HH}=\Id_{L_R}$ holds only when $\Gm$ is \emph{row-isometric}, i.e., when its rows form an orthonormal set. Throughout the paper, we assume that $\Gm$ has full row rank $L_R$, so that $\Rm_{\Gm}$ is invertible.

The unknown real parameter vector is
\begin{equation}
\label{eq:xi}
\vect{\xi}:=[\theta_1,\ldots,\theta_K,\,\Re\{\beta_1\},\Im\{\beta_1\},\ldots,\Re\{\beta_K\},\Im\{\beta_K\}]^{\TT}\in\R^{3K},\notag
\end{equation}
and the waveform energy $\|\sv\|^{2}:=\sum_{t=1}^{T}|s_t|^2$ is known. The \ac{ULA} is adopted for concreteness. As detailed in Remark~\ref{rem:upa}, every theoretical result of this paper, except the explicit \ac{ULA} aperture-scaling constants of Proposition~\ref{prop:ula-scaling}, extends to the planar array with only notational changes. 

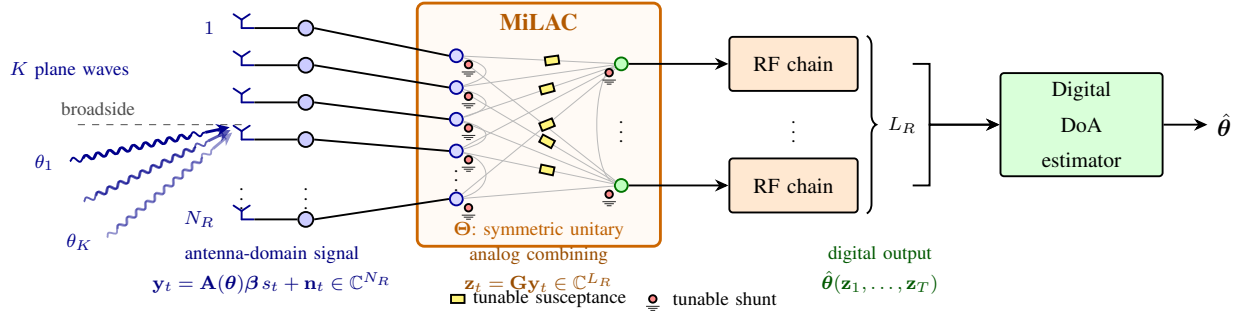
\begin{figure}[t]
\centering
\resizebox{1.0\linewidth}{!}{\begin{tikzpicture}[
  font=\footnotesize,
  >={Stealth[length=2mm,width=1.8mm]},
  ant/.style={rectangle, draw=none, fill=none,
              minimum width=3.6mm, minimum height=4.5mm, inner sep=0pt,
              path picture={%
                \draw[blue!55!black, thick, line cap=round]
                  ([xshift=-1.2mm, yshift=4.3mm]path picture bounding box.south)
                  -- ([yshift=3.5mm]path picture bounding box.south);
                \draw[blue!55!black, thick, line cap=round]
                  ([xshift=1.2mm, yshift=4.3mm]path picture bounding box.south)
                  -- ([yshift=3.5mm]path picture bounding box.south);
                \draw[blue!55!black, thick, line cap=round]
                  ([yshift=3.5mm]path picture bounding box.south)
                  -- ([yshift=2.25mm]path picture bounding box.south);
                \draw[blue!55!black, thick, line cap=round]
                  ([yshift=2.25mm]path picture bounding box.south)
                  -- ([xshift=1.8mm, yshift=2.25mm]path picture bounding box.south);%
              }},
  aport/.style={circle, draw, thick, fill=blue!20,
                minimum size=2.2mm, inner sep=0pt},
  iport/.style={circle, draw=blue!60!black, thick, fill=blue!15,
                minimum size=1.9mm, inner sep=0pt},
  rport/.style={circle, draw=green!45!black, thick, fill=green!25,
                minimum size=1.9mm, inner sep=0pt},
  susc/.style={draw, thick, fill=yellow!65, rectangle,
               minimum width=2.0mm, minimum height=1.2mm, inner sep=0pt},
  shunt/.style={draw, thick, fill=red!45, circle,
                minimum size=1.1mm, inner sep=0pt},
  gnd/.style={draw, black!80, line width=0.45pt, line cap=butt},
  shgnd/.pic={%
    \node[shunt] at (0,0) {};
    \draw[gnd] (-0.10,-0.11) -- (0.10,-0.11);
    \draw[gnd] (-0.07,-0.14) -- (0.07,-0.14);
    \draw[gnd] (-0.035,-0.17) -- (0.035,-0.17);
  },
  rfblock/.style={rectangle, draw, thick, rounded corners=1.5pt,
                  minimum height=8mm, minimum width=19mm,
                  align=center, fill=orange!18},
  digblock/.style={rectangle, draw, thick, rounded corners=1.5pt,
                   align=center, fill=green!15,
                   minimum height=13mm, minimum width=24mm},
  milacbox/.style={rectangle, draw=orange!80!black, very thick, rounded corners=3pt,
                   fill=orange!4, minimum width=36mm, minimum height=36mm},
  wave/.style={blue!55!black, thick, decorate,
               decoration={snake, amplitude=0.4mm, segment length=1.8mm,
                           post length=0mm, pre length=0mm}}
]

\def\dy{0.55}
\foreach \i in {1,...,4}{
  \node[ant]   (a\i) at (0,   {-\dy*(\i-1)}) {};
  \node[aport] (p\i) at (0.95,{-\dy*(\i-1)}) {};
  \draw[thick] (a\i.east) -- (p\i.west);
}
\node at (0.00,{-\dy*4.40}) {\scriptsize $\vdots$};
\node at (0.95,{-\dy*4.40}) {\scriptsize $\vdots$};
\node[ant]   (aN) at (0,   {-\dy*5.15}) {};
\node[aport] (pN) at (0.95,{-\dy*5.15}) {};
\draw[thick] (aN.east) -- (pN.west);

\node[left=0.8mm of a1, font=\scriptsize, blue!50!black] {$1$};
\node[left=0.8mm of aN, font=\scriptsize, blue!50!black] {$N_R$};

\coordinate (mid) at (0,{-\dy*2.60});

\draw[dashed, gray!60!black] (mid) -- ++(-2.5,0)
     node[above, font=\scriptsize, gray!60!black, pos=0.85] {broadside};

\foreach \ang/\opc in {12/1.0, 26/0.65, 40/0.42}{
  \begin{scope}[blue!55!black, opacity=\opc]
    \foreach \off in {0, 0.20}{
      \draw[wave]
        ($(mid)+({-(2.40+\off)*cos(\ang)},{-(2.40+\off)*sin(\ang)})$)
        -- ($(mid)+({-(0.55+\off)*cos(\ang)},{-(0.55+\off)*sin(\ang)})$);
    }
    \draw[->, very thick]
      ($(mid)+({-0.55*cos(\ang)},{-0.55*sin(\ang)})$)
      -- ($(mid)+({-0.18*cos(\ang)},{-0.18*sin(\ang)})$);
  \end{scope}
}

\node[blue!55!black, font=\scriptsize, anchor=east]
  at ($(mid)+({-2.70*cos(12)},{-2.70*sin(12)})$) {$\theta_1$};
\node[blue!55!black, font=\scriptsize, anchor=east]
  at ($(mid)+({-2.70*cos(40)},{-2.70*sin(40)})$) {$\theta_K$};

\node[blue!55!black, font=\scriptsize] at ($(mid)+(-2.55,0.80)$) {$K$ plane waves};

\node[milacbox] (M) at (4.40,{-\dy*2.60}) {};
\node[font=\small\bfseries, orange!70!black] at (M.north) [yshift=-3.0mm]
     {MiLAC};

\foreach \i/\yi in {1/1.02, 2/0.55, 3/0.08, 4/-0.39}{
  \node[iport] (ip\i) at ($(M.west)+(0.60,\yi)$) {};
}
\node at ($(M.west)+(0.60,-0.72)$) {\scriptsize $\vdots$};
\node[iport] (ipN) at ($(M.west)+(0.60,-1.10)$) {};

\node[rport] (op1) at ($(M.east)+(-0.60, 0.90)$) {};
\node at ($(M.east)+(-0.60, 0.00)$) {\scriptsize $\vdots$};
\node[rport] (opL) at ($(M.east)+(-0.60,-0.90)$) {};

\foreach \i in {1,2,3,4,N}{
  \foreach \j in {1,L}{
    \draw[gray!55, line width=0.35pt] (ip\i) -- (op\j);
  }
}
\draw[gray!55, line width=0.35pt] (ip1) .. controls +(0.55,-0.15) and +(0.55, 0.15) .. (ip2);
\draw[gray!55, line width=0.35pt] (ip2) .. controls +(0.55,-0.15) and +(0.55, 0.15) .. (ip3);
\draw[gray!55, line width=0.35pt] (ip3) .. controls +(0.55,-0.15) and +(0.55, 0.15) .. (ip4);
\draw[gray!55, line width=0.35pt] (ip4) .. controls +(0.55,-0.15) and +(0.55, 0.15) .. (ipN);
\draw[gray!55, line width=0.35pt] (op1) .. controls +(-0.45,-0.15) and +(-0.45, 0.15) .. (opL);

\node[susc, rotate=  8] at ($(ip1)!0.58!(op1)$) {};
\node[susc, rotate=-28] at ($(ip2)!0.55!(opL)$) {};
\node[susc, rotate= 15] at ($(ip3)!0.55!(op1)$) {};
\node[susc, rotate=-12] at ($(ip4)!0.55!(opL)$) {};
\node[susc, rotate= 22] at ($(ipN)!0.55!(op1)$) {};

\foreach \p in {ip1,ip2,ip3,ip4,ipN}{
  \pic at ($(\p)+(0.18,-0.13)$) {shgnd};
}
\foreach \p in {op1,opL}{
  \pic at ($(\p)+(-0.18,-0.13)$) {shgnd};
}

\draw[thick] (p1.east) -- (ip1.west);
\draw[thick] (p2.east) -- (ip2.west);
\draw[thick] (p3.east) -- (ip3.west);
\draw[thick] (p4.east) -- (ip4.west);
\draw[thick] (pN.east) -- (ipN.west);

\node[font=\scriptsize, orange!75!black]
      at ($(M.south)+(0,2.6mm)$)
      {$\bm{\Theta}$: symmetric unitary};

\node[rfblock] (rf1) at ($(op1)+(2.55,0)$)
     {RF chain};
\node at ($(op1)!0.5!(opL)+(2.55,0)$) {\scriptsize $\vdots$};
\node[rfblock] (rfL) at ($(opL)+(2.55,0)$)
     {RF chain};
\draw[thick, ->] (op1.east) -- (rf1.west);
\draw[thick, ->] (opL.east) -- (rfL.west);

\draw[decorate, decoration={brace, amplitude=4pt}, thick]
     ($(rf1.north east)+(0.12,0)$) -- ($(rfL.south east)+(0.12,0)$)
     node[midway, right=5pt, font=\scriptsize] {$L_R$};

\node[digblock] (dig) at ($(rf1.east |- M.center)+(3.30,0)$)
     {Digital\\[-0.2ex]DoA\\[-0.2ex]estimator};
\draw[thick, ->] ($(rf1.east)+(0.80,0)$) -- ++(0.25,0) |- (dig.west);
\draw[thick, ->] ($(rfL.east)+(0.80,0)$) -- ++(0.25,0) |- (dig.west);
\draw[thick, ->] (dig.east) -- ++(0.65,0)
     node[right, font=\small] {$\hat{\vect{\theta}}$};

\pgfmathsetmacro{\siglabel}{-\dy*6.20}
\node[font=\scriptsize, blue!50!black, align=center] at (0.45,\siglabel)
  {\\[0.5ex]antenna-domain signal\\[-0.3ex]
   $\mathbf{y}_t=\mat{A}(\vect{\theta})\vect{\beta}\,s_t+\mat{n}_t\in\C^{N_R}$};
\node[font=\scriptsize, orange!60!black, align=center] at (4.40,\siglabel)
  {\\[0.5ex]analog combining\\[-0.3ex]
   $\mathbf{z}_t=\mat{G}\mathbf{y}_t\in\C^{L_R}$};
\node[font=\scriptsize, green!35!black, align=center] at (9.45,\siglabel)
  {\\[0.5ex]digital output\\[-0.3ex]
   $\hat{\vect{\theta}}(\mathbf{z}_1,\dots,\mathbf{z}_T)$};

\begin{scope}[shift={(3.2,{\siglabel-0.60})}, font=\scriptsize]
  \node[susc]  (ls)  at (0,0) {};
  \node[right=2pt of ls, inner sep=1pt] (lstxt) {tunable susceptance};
  \pic (lsh) at ($(lstxt.east)+(0.35,0)$) {shgnd};
  \node[right=3pt of lstxt, xshift=5mm, inner sep=1pt] {tunable shunt};
\end{scope}

\end{tikzpicture}}
\caption{Receiver architecture studied in this paper. The fully-digital baseline ($\Gm=\Id_{N_R}$) corresponds to replacing the \ac{MiLAC} by a direct pass-through.}
\label{fig:system}
\end{figure}

\subsection{Lossless reciprocal MiLAC and feasibility set}
\label{subsec:whatismilac}
A \ac{MiLAC} is a reconfigurable multiport linear microwave network whose tunable internal components can be set so that the network realizes a desired linear transformation on the \ac{RF} signal between its ports~\cite{Nerini2025AnalogI}. In every existing communication-oriented design, two hardware-driven constraints are adopted~\cite{Nerini2025Capacity,Nerini2025Reduced,Wu2026MiLAC}. First, \emph{losslessness} forbids nonzero real parts of the tunable admittances because positive real parts dissipate signal energy and raise the receiver noise figure, while negative real parts require active devices with their own bias network and power budget. Second, \emph{reciprocity} forces the admittance matrix of the \ac{MiLAC} to be symmetric, which restricts the tunable branches to ordinary two-terminal reactive components (varactors, inductors, transmission-line stubs) rather than specialized non-reciprocal devices such as circulators or isolators.

Under the combination of losslessness and reciprocity, the scattering matrix $\Thetam$ of an $(N_R{+}L_R)$-port \ac{MiLAC}, taken at the standard reference impedance, is \emph{symmetric unitary}~\cite{Nerini2025Capacity}:
\begin{equation}
\label{eq:milac-constraints}
\Thetam^{\HH}\Thetam=\Id,\qquad \Thetam^{\TT}=\Thetam.
\end{equation}
Throughout the paper, ``\ac{MiLAC}'' refers to a lossless reciprocal multiport satisfying~\eqref{eq:milac-constraints}. Partitioning the scattering matrix into blocks aligned with the $N_R$ antenna ports and the $L_R$ \ac{RF}-chain ports,
\begin{equation}
\label{eq:theta-partition}
\Thetam=\begin{bmatrix}\Thetam_{11}&\Thetam_{12}\\\Thetam_{21}&\Thetam_{22}\end{bmatrix},
\end{equation}
where the first $N_R$ ports are terminated by the antennas and the last $L_R$ ports drive the downconversion chain, reciprocity forces $\Thetam_{12}=\Thetam_{21}^{\TT}$, and the effective $L_R\times N_R$ combiner in~\eqref{eq:obs} is $\Gm:=\Thetam_{21}$. The feasible set of physically realizable combiners is
\begin{equation}
\label{eq:feasible-set}
\begin{aligned}
\mathcal{F}_{\MiLAC}:=\bigl\{\Gm\in\C^{L_R\times N_R}:\;&\Gm=\Thetam_{21},\eqref{eq:milac-constraints}\bigr\}.
\end{aligned}
\end{equation}
Unitarity implies that the singular values of every $\Gm\in\mathcal{F}_{\MiLAC}$ lie in $[0,1]$. A cornerstone result of Section~\ref{sec:crb-general}, Lemma~\ref{lem:reachability}, shows that \emph{every} row-isometric matrix ($\Gm\Gm^{\HH}=\Id_{L_R}$) belongs to $\mathcal{F}_{\MiLAC}$, and that this is the only property of the feasibility set needed for the \ac{CRB} derivations that follow.

\subsection{Fully-digital baseline}
\label{subsec:digital-baseline}

The fully-digital receiver corresponds to $\Gm=\Id_{N_R}$ and serves as the benchmark against which every \ac{MiLAC} architecture in the sequel is compared. We derive its \ac{FIM} explicitly so that Section~\ref{sec:crb-general} can build the \ac{MiLAC} \ac{FIM} on top of it.

Let $\vect{\mu}_t:=\Am(\vect{\theta})\vect{\beta}\,s_t\in\C^{N_R}$ denote the noiseless antenna-domain mean of~\eqref{eq:antenna-obs}. Its partial derivatives with respect to $\vect{\xi}$ are
\begin{equation}
\label{eq:D-columns}
\begin{aligned}
\partial_{\theta_k}\vect{\mu}_t&=s_t\beta_k\dot{\av}(\theta_k),\\
\partial_{\Re\{\beta_k\}}\vect{\mu}_t&=s_t\av(\theta_k),\quad \partial_{\Im\{\beta_k\}}\vect{\mu}_t=js_t\av(\theta_k),
\end{aligned}
\end{equation}
where $\av(\theta)\in\C^{N_R}$ is the array steering vector and $\dot{\av}(\theta):=\partial_{\theta}\av(\theta)$ its angular derivative. Stacking the per-snapshot Jacobian matrices $\partial\vect{\mu}_t/\partial\vect{\xi}\in\C^{N_R\times 3K}$ vertically over $t=1,\ldots,T$ defines the $TN_R\times 3K$ digital Jacobian
\begin{equation}
\label{eq:Jdig-jacobian}
\mathcal{J}_{\dig}\,:=\,\bigl[\,(\partial\vect{\mu}_1/\partial\vect{\xi})^{\TT}\;\cdots\;(\partial\vect{\mu}_T/\partial\vect{\xi})^{\TT}\,\bigr]^{\TT}\in\C^{TN_R\times 3K}.
\end{equation}
Since $\yv_t\sim\mathcal{CN}(\vect{\mu}_t,\sigma^{2}\Id_{N_R})$ has parameter-independent covariance, the Slepian--Bangs formula~\cite{Stoica2005} yields the $3K\times 3K$ fully-digital \ac{FIM}
\begin{equation}
\label{eq:Jdig-def}
\Jm_{\dig}(\vect{\xi})\,=\,\tfrac{2}{\sigma^{2}}\Re\bigl\{\mathcal{J}_{\dig}^{\HH}\mathcal{J}_{\dig}\bigr\}.
\end{equation}
The $\vect{\theta}$-marginal \ac{CRB} $\crb_{\dig}(\vect{\theta})$ is the $K\times K$ bound on the target angles alone. Since the complex amplitudes $(\Re\{\beta_k\},\Im\{\beta_k\})$ are unknown nuisance parameters, this bound accounts for the cost of estimating them jointly with the angles by taking the Schur complement of the amplitude block of $\Jm_{\dig}(\vect{\xi})$ before inverting. The resulting bound lives in the L\"owner order on \ac{PSD} matrices. Section~\ref{sec:crb-general} shows that $\Jm_{\dig}$ upper-bounds the \ac{MiLAC} \ac{FIM} in the L\"owner sense for every feasible $\Gm$, with equality characterized by a simple subspace containment.

\subsection{A toy example: why sensing differs from communication}
\label{subsec:toy}

Before developing the general theory, we use the simplest possible case, a single target ($K{=}1$) observed by an $N_R$-element array, to build intuition for two questions a reader new to \ac{MiLAC} naturally asks: why does combiner design for \emph{sensing} differ from combiner design for \emph{communication}, and what does a \ac{MiLAC} offer us in the sensing setting?

\emph{Communication:} Suppose the array carries one data stream arriving from a known direction $\theta$ along the steering vector $\av(\theta)$. A receiver that only needs to recover the stream maximizes the post-combining \ac{SNR}, which a single \ac{RF} chain already achieves by matched filtering, $\gv^{\HH}\propto\av(\theta)^{\HH}$. The only direction that matters is $\av(\theta)$ itself: a one-dimensional combiner row space suffices, and any combiner whose row space contains $\av(\theta)$ is optimal.

\emph{Sensing:} Now suppose $\theta$ is \emph{unknown} and is precisely the quantity to be estimated. What makes $\theta$ estimable is not the value of $\av(\theta)$ but how the observation \emph{changes} as $\theta$ varies, that is, the angular derivative $\dot{\av}(\theta)$. The matched-filter combiner $\gv^{\HH}\propto\av(\theta)^{\HH}$, optimal for communication, collapses the array onto the single direction $\av(\theta)$ and so retains no separate sensitivity to the variation $\dot{\av}(\theta)$ that carries the angle. Through it the \ac{CRB} on $\theta$ is in fact infinite, since a single complex measurement per snapshot cannot resolve the three real unknowns $(\theta,\Re\{\beta\},\Im\{\beta\})$. Equation~\eqref{eq:D-columns} makes this precise: the noiseless mean responds to these unknowns only along the two directions $\av(\theta)$, which carries the amplitude, and $\dot{\av}(\theta)$, which carries the angle. Since the combiner retains only the component of each direction lying in its row space, preserving all the Fisher information requires that row space to contain \emph{both}, i.e., a two-dimensional subspace. This is the single-target instance of the ``two \ac{RF} chains per target'' rule and the reason a communication-optimal front end is not automatically a sensing-optimal one.

\emph{What \ac{MiLAC} offers:} A lossless reciprocal \ac{MiLAC} can realize \emph{any} combiner whose rows are orthonormal (Lemma~\ref{lem:reachability}), so it can place its two \ac{RF} chains exactly on $\mathrm{span}\{\av(\theta),\dot{\av}(\theta)\}$ and thereby recover the \emph{entire} Fisher information of a fully-digital $N_R$-chain receiver with only two active chains, a saving that grows with the array size. A constant-modulus phase-shifter combiner with the same two chains cannot freely orient its row space and, as shown in Section~\ref{subsec:ps-strict}, generically misses this two-dimensional subspace. Sections~\ref{sec:crb-general} and~\ref{sec:hardware} make these statements precise and extend them to $K$ targets.

\section{CRB of DoA Estimation with MiLAC}
\label{sec:crb-general}

This section develops the $K$-target \ac{CRB} theory under arbitrary \ac{MiLAC} combining: a row-space-projector form of the \ac{FIM} (Section~\ref{subsec:fim-arbitrary}), the resulting Grassmannian descent (Section~\ref{subsec:subspace}), the L\"owner ordering and zero-gap theorem (Section~\ref{subsec:lowner}), row-isometry reachability (Section~\ref{subsec:reachability}), the \ac{ULA} aperture-scaling law (Section~\ref{subsec:ula-scaling}), and the strict inferiority of phase-shifter combining (Section~\ref{subsec:ps-strict}).

\subsection{Fisher information under arbitrary MiLAC combining}
\label{subsec:fim-arbitrary}

We now lift the fully-digital \ac{FIM} of Section~\ref{subsec:digital-baseline} to an arbitrary lossless reciprocal \ac{MiLAC} combiner. The post-combining observation $\zv_t=\Gm\vect{\mu}_t+\Gm\nv_t$ of~\eqref{eq:obs} has mean $\Gm\vect{\mu}_t$ and parameter-independent covariance $\sigma^{2}\Rm_{\Gm}$. Its Jacobian with respect to $\vect{\xi}$ is the digital Jacobian~\eqref{eq:Jdig-jacobian} left-multiplied by $\Id_T\otimes\Gm$, namely
\begin{equation}
\label{eq:JMiLAC-jacobian}
\mathcal{J}_{\MiLAC}\,:=\,(\Id_T\otimes\Gm)\,\mathcal{J}_{\dig}\,\in\,\C^{TL_R\times 3K}.
\end{equation}
Since $\zv_t\sim\mathcal{CN}(\Gm\vect{\mu}_t,\sigma^{2}\Rm_{\Gm})$ has parameter-independent covariance, the same Slepian--Bangs formula used in~\eqref{eq:Jdig-def}, now with $\sigma^{2}\Rm_{\Gm}$ in place of $\sigma^{2}\Id_{N_R}$, yields
\begin{equation}
\label{eq:SB-raw}
\Jm_{\MiLAC}(\vect{\xi};\Gm)=\tfrac{2}{\sigma^{2}}\Re\bigl\{\mathcal{J}_{\MiLAC}^{\HH}(\Id_T\otimes\Rm_{\Gm}^{-1})\mathcal{J}_{\MiLAC}\bigr\}.
\end{equation}
Substituting~\eqref{eq:JMiLAC-jacobian} into~\eqref{eq:SB-raw} and using the Kronecker identity
\begin{equation}
\label{eq:kron-body}
(\Id_T\otimes\Gm)^{\HH}(\Id_T\otimes\Rm_{\Gm}^{-1})(\Id_T\otimes\Gm)=\Id_T\otimes(\Gm^{\HH}\Rm_{\Gm}^{-1}\Gm),
\end{equation}
together with the Moore--Penrose projector identity~\cite{HornJohnson2013}
\begin{equation}
\label{eq:projector}
\Gm^{\HH}\Rm_{\Gm}^{-1}\Gm\,=\,\Pm_{\Gm},\qquad \Pm_{\Gm}:=\Pm_{\row(\Gm)},
\end{equation}
collapses the $\Gm$-dependence of $\Jm_{\MiLAC}$ entirely onto the row-space projector $\Pm_{\Gm}$.

\begin{theorem}[\ac{MiLAC} \ac{FIM} in projector form]
\label{thm:fim-projector}
Under~\eqref{eq:obs} and $\rank(\Gm)=L_R$,
\begin{equation}
\label{eq:SB-proj}
\Jm_{\MiLAC}(\vect{\xi};\Gm)\,=\,\tfrac{2}{\sigma^{2}}\Re\bigl\{\mathcal{J}_{\dig}^{\HH}(\Id_T\otimes\Pm_{\Gm})\mathcal{J}_{\dig}\bigr\},
\end{equation}
which recovers the fully-digital benchmark~\eqref{eq:Jdig-def} at $\Pm_{\Gm}=\Id_{N_R}$.
\end{theorem}

\begin{IEEEproof}
See Appendix~\ref{app:fim-projector}.
\end{IEEEproof}

In words, Theorem~\ref{thm:fim-projector} says that, as far as Fisher information is concerned, the only thing that matters about a \ac{MiLAC} combiner $\Gm$ is the \emph{subspace its rows span}, not the individual entries of $\Gm$. Two combiners with completely different entries but the same row space extract exactly the same information about the targets, and the fully-digital receiver is recovered precisely when that row space is all of $\C^{N_R}$. Combiner design therefore reduces to choosing a good subspace, an observation we exploit repeatedly below.

\subsection{Subspace invariance and Grassmannian descent}
\label{subsec:subspace}

A priori, $\Jm_{\MiLAC}$ is a matrix-valued function on the manifold of full-row-rank matrices in $\C^{L_R\times N_R}$, which has real dimension $2L_RN_R$. The projector identity~\eqref{eq:projector} immediately reveals that this dependence is only through $\Pm_{\Gm}$, which is uniquely determined by $\row(\Gm)\in\mathrm{Gr}(L_R,N_R)$.

\begin{proposition}[Subspace invariance and Grassmannian descent]
\label{prop:subspace}
Let $\Jm_{\MiLAC}(\vect{\xi};\Gm)$ be defined as in Theorem~\ref{thm:fim-projector}. Then:
\begin{enumerate}
\item[\textup{(i)}] \emph{Subspace invariance:} $\row(\Gm)=\row(\Gm')$ implies $\Jm_{\MiLAC}(\vect{\xi};\Gm)=\Jm_{\MiLAC}(\vect{\xi};\Gm')$.
\item[\textup{(ii)}] \emph{Row-isometric representation:} Every $\mathcal{S}\in\mathrm{Gr}(L_R,N_R)$ is the row space of some $\Gm_{\mathrm{iso}}\in\mathrm{St}(L_R,N_R)$.
\end{enumerate}
\end{proposition}

\begin{IEEEproof}
See Appendix~\ref{app:subspace}.
\end{IEEEproof}

Proposition~\ref{prop:subspace} collapses the optimization domain from complex matrices (real dimension $2L_RN_R$) to the Grassmannian (real dimension $2L_R(N_R-L_R)$). By~(i) and~(ii) the achievable Fisher information is the same whether $\Gm$ ranges over all full-row-rank matrices or only row-isometric ones, so throughout the sequel we restrict, without loss of generality, every optimization of $\Jm_{\MiLAC}$ (or $\crb_{\MiLAC}$) to row-isometric $\Gm$, or equivalently to subspaces $\mathcal{S}\in\mathrm{Gr}(L_R,N_R)$.

\subsection{Löwner ordering, zero-gap theorem, and identifiability}
\label{subsec:lowner}

We now show that the \ac{MiLAC} \ac{FIM} can only lose, never gain, Fisher information relative to its fully-digital counterpart, and identify the precise row-space condition under which the loss is zero. Define the joint steering--derivative subspace
\begin{equation}
\label{eq:S-K-star}
\mathcal{S}_{K}^{\star}(\vect{\theta})\,:=\,\mathrm{span}_{\C}\{\av(\theta_k),\dot{\av}(\theta_k)\}_{k=1}^{K}\,\subseteq\,\C^{N_R}.
\end{equation}
For distinct target angles and $2K\le N_R$, $\dim\mathcal{S}_{K}^{\star}(\vect{\theta})=2K$: this follows from a Vandermonde argument for the \ac{ULA} and from a generic linear-independence argument for unstructured arrays. Concretely, the entries of $\av(\theta)$ are successive powers of $e^{j\pi\sin\theta}$, so steering vectors at distinct angles are linearly independent, and adjoining their derivatives $\dot{\av}(\theta_k)$ keeps all $2K$ vectors independent. This is the standard non-degeneracy condition assumed in \ac{DoA} estimation~\cite{Stoica2005}. We assume this generic non-degeneracy throughout.

\begin{theorem}[\ac{FIM} ordering and equality condition]
\label{thm:lowner}
Let $\Gm\in\C^{L_R\times N_R}$ be full row rank and $\mathcal{S}:=\row(\Gm)$. Then:
\begin{enumerate}
\item[\textup{(i)}] $\Jm_{\MiLAC}(\vect{\xi};\Gm)\preceq\Jm_{\dig}(\vect{\xi})$.
\item[\textup{(ii)}] Whenever the marginalized $\vect\theta$-block \acp{FIM} of both the \ac{MiLAC} and digital systems are positive definite, $\crb_{\MiLAC}(\vect{\theta};\Gm)\succeq\crb_{\dig}(\vect{\theta})$ in the L\"owner sense on $K\times K$ matrices.
\item[\textup{(iii)}] Equality $\Jm_{\MiLAC}(\vect{\xi};\Gm)=\Jm_{\dig}(\vect{\xi})$ holds (as a matrix identity on $\R^{3K}$) if and only if $\mathcal{S}\supseteq\mathcal{S}_{K}^{\star}(\vect{\theta})$.
\end{enumerate}
\end{theorem}

\begin{IEEEproof}
See Appendix~\ref{app:lowner}.
\end{IEEEproof}

\begin{corollary}[Zero-gap threshold]
\label{cor:zerogap}
Assume $2K\le N_R$ and $\dim\mathcal{S}_{K}^{\star}(\vect{\theta})=2K$. For every $L_R\ge 2K$ and every $\Gm\in\mathcal{F}_{\MiLAC}$ with $\row(\Gm)\supseteq\mathcal{S}_{K}^{\star}(\vect{\theta})$,
\begin{equation*}
\Jm_{\MiLAC}(\vect{\xi};\Gm)=\Jm_{\dig}(\vect{\xi}),\;\crb_{\MiLAC}(\vect{\theta};\Gm)=\crb_{\dig}(\vect{\theta}).
\end{equation*}
\end{corollary}

\begin{remark}[Identifiability versus zero-gap thresholds]
\label{rem:thresholds}
The $K$-target problem with a shared waveform admits two distinct thresholds on $L_R$.

\emph{(a) Identifiability threshold $L_R\ge\lceil 3K/2\rceil$:} Because the waveform $s_t$ is a known scalar, the Jacobian factorizes as
\begin{equation}
\label{eq:jdig-factor}
\mathcal{J}_{\dig}=\sv\otimes\Nm_{0},\qquad \Nm_{0}\in\C^{N_R\times 3K},
\end{equation}
with $\Nm_{0}$ collecting the columns $(\beta_k\dot{\av}(\theta_k),\av(\theta_k),j\av(\theta_k))_{k=1}^{K}$. Substituting into~\eqref{eq:SB-proj} and writing $\Gm_{\mathrm{iso}}\Nm_{0}=\Nm_{r}+j\Nm_{i}$, the \ac{FIM} reduces to
\begin{equation}
\label{eq:fim-rank}
\frac{2\|\sv\|^{2}}{\sigma^{2}}\bigl(\Nm_{r}^{\TT}\Nm_{r}+\Nm_{i}^{\TT}\Nm_{i}\bigr),
\end{equation}
whose rank is upper-bounded by $2L_R$. Identifiability of the $3K$-dimensional real parameter vector therefore requires $L_R\ge\lceil 3K/2\rceil$, and below this threshold $\crb_{\MiLAC}$ is infinite.

\emph{(b) Zero-gap threshold $L_R\ge 2K$:} This threshold comes directly from Theorem~\ref{thm:lowner}(iii) via Corollary~\ref{cor:zerogap}. In the intermediate regime $\lceil 3K/2\rceil\le L_R<2K$ the problem is identifiable but the \ac{MiLAC} \ac{CRB} sits strictly above the digital benchmark. The two thresholds coincide at $K=1$ (both equal to $2$) and first differ at $K=2$ (taking values $3$ and $4$). Section~\ref{sec:numerical} exhibits the resulting three-regime structure numerically at $K=3$ (thresholds $5$ and $6$), where the intermediate regime sits well clear of the unidentifiable boundary.
\end{remark}

\subsection{Row-isometry reachability}
\label{subsec:reachability}

The next ingredient closes the loop between information-theoretic optimality and physical realizability.

\begin{lemma}[Row-isometry reachability]
\label{lem:reachability}
Every row-isometric $\Gm_{0}\in\mathrm{St}(L_R,N_R)$ belongs to $\mathcal{F}_{\MiLAC}$. In particular, for $L_R\ge 2K$, the row-isometric matrix $\Gm_{K}^{\star}\in\mathrm{St}(L_R,N_R)$ obtained by appending $L_R{-}2K$ orthonormal complement rows to the Gram--Schmidt orthonormalization of $\{\av(\theta_k),\dot{\av}(\theta_k)\}_{k=1}^{K}$ lies in $\mathcal{F}_{\MiLAC}$ and satisfies $\row(\Gm_{K}^{\star})\supseteq\mathcal{S}_{K}^{\star}(\vect{\theta})$.
\end{lemma}

\begin{IEEEproof}
The first claim is the row-isometric case of the symmetric-unitary completion result developed in~\cite{Nerini2025Capacity} and crisply stated as Proposition~1 of~\cite{Wu2026MiLAC}. That proposition exhibits, for any matrix $\Xm$ with $\|\Xm\|_2\le 1$, a symmetric unitary $\Thetam_{0}\in\C^{(N_R+L_R)\times(N_R+L_R)}$ whose off-diagonal block is $\Xm$, by populating the two diagonal blocks with terms involving $\sqrt{1-\sigma_i^{2}(\Xm)}$, where $\sigma_i(\Xm)$ are the singular values of $\Xm$. Specializing to $\Xm=\Gm_{0}$, row-isometry makes every $\sigma_i(\Gm_{0})$ equal one, so these terms vanish: the $L_R\times L_R$ block of $\Thetam_{0}$ reduces to $\mathbf{0}$ and the $N_R\times N_R$ block reduces to $\overline{\Vm}\Lambdam^{1/2}\Vm^{\HH}$, where $\Vm\Lambdam\Vm^{\HH}=\Id_{N_R}-\Gm_{0}^{\HH}\Gm_{0}$ is the spectral decomposition of the kernel projector (the orthogonal projector onto $\ker(\Gm_{0})$). The resulting completion is
\begin{equation*}
\Thetam_{0}=\begin{bmatrix}\overline{\Vm}\Lambdam^{1/2}\Vm^{\HH} & \Gm_{0}^{\TT}\\ \Gm_{0} & \mathbf{0}_{L_R\times L_R}\end{bmatrix},
\end{equation*}
so the realization reduces to a single $N_R\times N_R$ Hermitian eigendecomposition, with no $L_R\times L_R$ SVD branch needed. The second claim follows from Proposition~\ref{prop:subspace}(ii) applied to the Gram--Schmidt basis of $\mathcal{S}_{K}^{\star}(\vect{\theta})$ padded with $L_R-2K$ orthonormal complement rows.
\end{IEEEproof}

\subsection{ULA aperture scaling}
\label{subsec:ula-scaling}

We now establish the per-target $\mathcal{O}(N_R^{-3})$ scaling that gives \ac{MiLAC}-aided sensing its asymptotic punch, with the explicit constant for the single-target case.

\begin{proposition}[\ac{ULA} aperture scaling]
\label{prop:ula-scaling}
For the half-wavelength $N_R$-element \ac{ULA} with $[\av(\theta)]_n=e^{j\pi n\sin\theta}$, $n=0,\ldots,N_R-1$, and $K\ge 1$ targets at distinct angles $\theta_1,\ldots,\theta_K$ independent of $N_R$, the per-target marginalized digital \ac{CRB} satisfies
\begin{equation}
\label{eq:ula-rate}
[\crb_{\dig}(\vect{\theta})]_{kk}\,=\,\mathcal{O}(N_R^{-3}),\qquad k=1,\ldots,K,
\end{equation}
as $N_R\to\infty$ with $K$, $\vect{\theta}$, and $\vect{\beta}$ fixed. By Corollary~\ref{cor:zerogap}, every \ac{CRB}-optimal \ac{MiLAC} with $L_R\ge 2K$ and $\row(\Gm)\supseteq\mathcal{S}_{K}^{\star}(\vect{\theta})$ matches this rate exactly.
\end{proposition}

\begin{IEEEproof}
See Appendix~\ref{app:ula-scaling}.
\end{IEEEproof}

\begin{remark}[$K{=}1$ closed-form constant]
\label{rem:ula-K1}
For $K=1$ the rate~\eqref{eq:ula-rate} admits the closed-form constant
\begin{equation}
\label{eq:ula-K1}
\crb^{\star}(\theta)\,=\,\frac{6\sigma^{2}}{|\beta|^{2}\|\sv\|^{2}\pi^{2}\cos^{2}\theta\,N_R(N_R^{2}-1)},
\end{equation}
matching the fully-digital cubic scaling and degrading as $1/\cos^{2}\theta$ near endfire, consistent with the loss of effective aperture at $\theta\to\pm 90^{\circ}$.
\end{remark}

\subsection{Strict inferiority of phase-shifter combining}
\label{subsec:ps-strict}

We close the section by linking the numerical observation that a phase-shifter combiner with $L_R=2K$ chains leaves a strictly positive gap to the digital \ac{CRB} to a geometric guarantee that goes through the same row-space machinery developed above.

A phase-shifter combiner is any matrix $\Fm\in\C^{L_R\times N_R}$ whose entries satisfy $|[\Fm]_{\ell,n}|=1/\sqrt{N_R}$ for all $(\ell,n)$, so that each row has unit norm. Let
\begin{equation}
\label{eq:ps-feas}
\mathcal{F}_{\mathrm{PS}}\,:=\,\bigl\{\Fm\in\C^{L_R\times N_R}:|[\Fm]_{\ell,n}|=\frac{1}{\sqrt{N_R}},\forall\ell,n\bigr\}
\end{equation}
denote the feasibility set, a real-analytic torus of real dimension $L_R N_R$. By Proposition~\ref{prop:subspace}, the phase-shifter combiner Fisher information depends on $\Fm$ only through the orthogonal projector $\Pm_{\row(\Fm)}$, so the attainable accuracy is controlled by the set of reachable row spaces

\begin{align}
\mathcal{R}_{\mathrm{PS}} &\,:=\,\{\row(\Fm):\Fm\in\mathcal{F}_{\mathrm{PS}},\,\rank(\Fm)=L_R\}\notag\\
&\subseteq\mathrm{Gr}(L_R,N_R).\label{eq:ps-rows}
\end{align}

\begin{proposition}[Strict inferiority of phase-shifter combining]
\label{prop:ps-strict}
Fix $K\ge 1$, the chain count $L_R\ge 2K$, and assume the mild aperture condition $N_R\ge 2L_R$. Then:
\begin{enumerate}
\item[\textup{(i)}] \emph{L\"owner ordering:} For every $\Fm\in\mathcal{F}_{\mathrm{PS}}$, $\Jm_{\mathrm{PS}}(\vect{\xi};\Fm)\preceq\Jm_{\dig}(\vect{\xi})$ and $\crb_{\mathrm{PS}}(\vect{\theta};\Fm)\succeq\crb_{\dig}(\vect{\theta})$, with equality if and only if $\row(\Fm)\supseteq\mathcal{S}_{K}^{\star}(\vect{\theta})$.
\item[\textup{(ii)}] \emph{Reach deficit on the Grassmannian:} The set $\mathcal{R}_{\mathrm{PS}}$ of subspaces reachable as row spaces of phase-shifter combiners has at most $L_R(N_R-1)$ real degrees of freedom, falling short of $\dim_{\R}\mathrm{Gr}(L_R,N_R)=2L_R(N_R-L_R)$ by
\begin{equation}
\label{eq:ps-codim}
\Delta(L_R,N_R)\,:=\,L_R(N_R-2L_R+1)\,>\,0.
\end{equation}
In probabilistic terms, an $L_R$-dimensional subspace drawn uniformly at random from $\mathrm{Gr}(L_R,N_R)$ almost surely (with probability one) falls outside $\mathcal{R}_{\mathrm{PS}}$, so $\mathcal{R}_{\mathrm{PS}}$ occupies vanishing volume in the Grassmannian.
\item[\textup{(iii)}] \emph{Generic strict gap:} Fix any phase-shifter combiner $\Fm\in\mathcal{F}_{\mathrm{PS}}$. Because the steering curve $\theta\mapsto\av(\theta)$ meets the fixed proper subspace $\row(\Fm)$ only at isolated angles, the target configurations satisfying $\row(\Fm)\supseteq\mathcal{S}_{K}^{\star}(\vect{\theta})$ form a measure-zero set. Hence, by~(i), $\crb_{\mathrm{PS}}(\vect{\theta};\Fm)\succ\crb_{\dig}(\vect{\theta})$ strictly for almost every $\vect{\theta}$, reducing for a single target to $\crb_{\mathrm{PS}}(\theta;\Fm)>\crb_{\dig}(\theta)$. A feasible \ac{MiLAC}, by contrast, can be tuned so that $\row(\Gm)\supseteq\mathcal{S}_{K}^{\star}(\vect{\theta})$ for every $\vect{\theta}$, attaining $\crb_{\dig}(\vect{\theta})$ exactly.
\end{enumerate}
\end{proposition}

\begin{IEEEproof}
See Appendix~\ref{app:ps-strict}.
\end{IEEEproof}

Proposition~\ref{prop:ps-strict} promotes the phase-shifter combining curves of Section~\ref{sec:numerical} from a numerical observation to a theoretical statement: a given phase-shifter design meets the equality condition $\row(\Fm)\supseteq\mathcal{S}_{K}^{\star}(\vect{\theta})$ only on a measure-zero set of target configurations, and therefore leaves a strictly positive \ac{CRB} gap for almost every geometry. By contrast, the \ac{MiLAC} feasibility set $\mathcal{F}_{\MiLAC}$ covers the entire Stiefel manifold (Lemma~\ref{lem:reachability}) and can be tuned onto $\mathcal{S}_{K}^{\star}(\vect{\theta})$ exactly through the Gram--Schmidt construction.

The geometric interpretation behind this result is intuitive. The candidate combiner's row spaces form the Grassmannian, and matching the digital \ac{CRB} means landing on the one subspace $\mathcal{S}_{K}^{\star}(\vect{\theta})$ that the targets dictate. A \ac{MiLAC} can realize \emph{every} orthonormal row space, i.e., \emph{it covers the entire Stiefel manifold}, so it can always be steered onto $\mathcal{S}_{K}^{\star}(\vect{\theta})$. A phase-shifter combiner, whose entries are locked to constant modulus, can reach only a thin sliver of all possible row spaces, a set of strictly smaller dimension that therefore occupies zero volume in the Grassmannian, much as a curve occupies zero area in a plane. Consequently, a fixed phase-shifter design lands exactly on the target-dictated subspace only for exceptional geometries, and otherwise leaves a strictly positive gap. The exceptions are non-generic configurations, for example targets whose steering and derivative directions align with \ac{DFT} beams. In such special geometries, or when only a coarse angle estimate is required, a fixed phase-shifter combiner can still be adequate.

\begin{remark}[Extension to planar arrays]
\label{rem:upa}
All results above extend from the \ac{ULA} to the \ac{UPA}. For the \ac{UPA}, with azimuth $\varphi$ and elevation $\vartheta$, the per-target steering--derivative subspace $\mathrm{span}_{\C}\{\av(\varphi_k,\vartheta_k),\partial_{\varphi}\av(\varphi_k,\vartheta_k),\partial_{\vartheta}\av(\varphi_k,\vartheta_k)\}$ has dimension three rather than two, so $\mathcal{S}_{K}^{\star}$ has generic dimension $3K$. Every theorem, proposition, and lemma in this section carries over with the substitution $2\to 3$ per target: the zero-gap threshold becomes $L_R\ge 3K$, the identifiability threshold becomes $L_R\ge 2K$, and the stem-connected cost of Section~\ref{sec:hardware} becomes $3K(2N_R+1)$ tunable susceptances. Only the explicit \ac{ULA} aperture-scaling constants of Remark~\ref{rem:ula-K1} are array-specific, while analogous \ac{UPA} expressions follow from the Kronecker structure $\av(\varphi,\vartheta)=\av_x(\varphi)\otimes\av_y(\vartheta)$. We develop the theory for the \ac{ULA} because it yields the cleanest design rules.
\end{remark}

\begin{figure}[t]
\centering
\resizebox{0.8\linewidth}{!}{\begin{tikzpicture}[
  font=\footnotesize,
  >={Stealth[length=1.8mm,width=1.6mm]},
  aport/.style={circle, draw, thick, fill=blue!18,
                minimum size=2.8mm, inner sep=0pt},
  rport/.style={circle, draw, thick, fill=green!25,
                minimum size=2.8mm, inner sep=0pt},
  edge/.style={draw, black!55, line width=0.55pt, line cap=round},
  centraledge/.style={draw, black!80, line width=0.75pt, line cap=round},
  fanedge/.style={draw, black!70, line width=0.55pt, line cap=round, dash pattern=on 2pt off 1pt},
  gnd/.style={draw, black!80, line width=0.55pt, line cap=butt},
  susc/.style={draw, thick, fill=yellow!70, rectangle,
               minimum width=2.3mm, minimum height=1.35mm,
               inner sep=0pt},
  shunt/.style={draw, thick, fill=red!55, circle,
                minimum size=1.35mm, inner sep=0pt},
  cgdash/.style={draw, black!55, dashed, thick, rounded corners=3pt},
  panelttl/.style={font=\small\bfseries},
  shgndD/.pic={%
    \node[shunt] at (0,0) {};
    \draw[gnd] (-0.13,-0.13) -- (0.13,-0.13);
    \draw[gnd] (-0.09,-0.17) -- (0.09,-0.17);
    \draw[gnd] (-0.045,-0.21) -- (0.045,-0.21);
  }
]

\begin{scope}[shift={(0,0)}]
  \node[panelttl] at (0,2.4) {(a) Fully-connected MiLAC};

  \def\R{1.35}

  \node[aport] (A1) at ({\R*cos(150)},{\R*sin(150)}) {};
  \node[aport] (A2) at ({\R*cos(210)},{\R*sin(210)}) {};
  \node[aport] (A3) at ({\R*cos(270)},{\R*sin(270)}) {};
  \node[aport] (A4) at ({\R*cos( 90)},{\R*sin( 90)}) {};
  \node[rport] (R1) at ({\R*cos( 30)},{\R*sin( 30)}) {};
  \node[rport] (R2) at ({\R*cos(330)},{\R*sin(330)}) {};

  \foreach \x/\y in {A1/A2,A1/A3,A1/A4,A1/R1,A1/R2,
                     A2/A3,A2/A4,A2/R1,A2/R2,
                     A3/A4,A3/R1,A3/R2,
                     A4/R1,A4/R2,
                     R1/R2}{
    \draw[edge] (\x) -- (\y);
  }
  \foreach \x/\y/\f in {A1/A2/0.42, A1/A3/0.42, A1/A4/0.50, A1/R1/0.42, A1/R2/0.42,
                        A2/A3/0.50, A2/A4/0.42, A2/R1/0.42, A2/R2/0.42,
                        A3/A4/0.42, A3/R1/0.42, A3/R2/0.42,
                        A4/R1/0.50, A4/R2/0.50,
                        R1/R2/0.50}{
    \node[susc] at ($(\x)!\f!(\y)$) {};
  }
  \foreach \n/\a in {A1/150, A2/210, A3/270, A4/90, R1/30, R2/330}{
    \draw[edge] (\n) -- ({(\R+0.22)*cos(\a)},{(\R+0.22)*sin(\a)});
    \pic[rotate={\a+90}] at ({(\R+0.2)*cos(\a)},{(\R+0.2)*sin(\a)}) {shgndD};
  }

  \node[font=\scriptsize] at ({(\R+0.78)*cos(150)},{(\R+0.78)*sin(150)}) {$a_2$};
  \node[font=\scriptsize] at ({(\R+0.78)*cos(210)},{(\R+0.78)*sin(210)}) {$a_3$};
  \node[font=\scriptsize] at ({(\R+0.78)*cos(270)},{(\R+0.78)*sin(270)}) {$a_4$};
  \node[font=\scriptsize] at ({(\R+0.78)*cos( 90)},{(\R+0.78)*sin( 90)}) {$a_1$};
  \node[font=\scriptsize] at ({(\R+0.78)*cos( 10)},{(\R+0.78)*sin( 10)}) {$r_1$};
  \node[font=\scriptsize] at ({(\R+0.78)*cos(330)},{(\R+0.78)*sin(330)}) {$r_2$};
\end{scope}

\begin{scope}[shift={(5.5,0)}]
  \node[panelttl] at (0, 2.4) {(b) Stem-connected MiLAC};

  \def\KR{0.55}
  \def\SR{1.95}

  \coordinate (Cr1) at ({\KR*cos( 90)}, {\KR*sin( 90)});
  \coordinate (Cr2) at ({\KR*cos(-30)}, {\KR*sin(-30)});
  \coordinate (Ca1) at ({\KR*cos(210)}, {\KR*sin(210)});

  \coordinate (Ca2) at ({\SR*cos(150)}, {\SR*sin(150)});
  \coordinate (Ca3) at ({\SR*cos( 30)}, {\SR*sin( 30)});
  \coordinate (Ca4) at ({\SR*cos(270)}, {\SR*sin(270)});

  \draw[cgdash]
    ({-\KR-0.20}, {-\KR-0.18})
    rectangle
    ({ \KR+0.20}, { \KR+0.18});

  \foreach \x/\y in {Cr1/Cr2, Cr2/Ca1, Ca1/Cr1}{
    \draw[centraledge] (\x) -- (\y);
    \node[susc] at ($(\x)!0.5!(\y)$) {};
  }

  \foreach \cv/\f in {Cr1/0.42, Cr2/0.50, Ca1/0.40}{
    \draw[fanedge] (Ca2) -- (\cv);
    \node[susc] at ($(Ca2)!\f!(\cv)$) {};
  }
  \foreach \cv/\f in {Cr1/0.42, Cr2/0.40, Ca1/0.50}{
    \draw[fanedge] (Ca3) -- (\cv);
    \node[susc] at ($(Ca3)!\f!(\cv)$) {};
  }
  \foreach \cv/\f in {Cr1/0.50, Cr2/0.42, Ca1/0.42}{
    \draw[fanedge] (Ca4) -- (\cv);
    \node[susc] at ($(Ca4)!\f!(\cv)$) {};
  }

  \node[rport] at (Cr1) {};
  \node[rport] at (Cr2) {};
  \node[aport] at (Ca1) {};
  \node[aport] at (Ca2) {};
  \node[aport] at (Ca3) {};
  \node[aport] at (Ca4) {};

  \def\LL{0.26}
  \def\PP{0.50}

  \draw[edge] (Cr1) -- ($(Cr1)+({\LL*cos( 90)},{\LL*sin( 90)})$);
  \pic[rotate=180] at ($(Cr1)+({\PP*cos( 90)},{\PP*sin( 90)})$) {shgndD};
  \draw[edge] (Cr2) -- ($(Cr2)+({\LL*cos(-30)},{\LL*sin(-30)})$);
  \pic[rotate=60]  at ($(Cr2)+({\PP*cos(-30)},{\PP*sin(-30)})$) {shgndD};
  \draw[edge] (Ca1) -- ($(Ca1)+({\LL*cos(210)},{\LL*sin(210)})$);
  \pic[rotate=300] at ($(Ca1)+({\PP*cos(210)},{\PP*sin(210)})$) {shgndD};

  \draw[edge] (Ca2) -- ($(Ca2)+({\LL*cos(150)},{\LL*sin(150)})$);
  \pic[rotate=240] at ($(Ca2)+({\PP*cos(150)},{\PP*sin(150)})$) {shgndD};
  \draw[edge] (Ca3) -- ($(Ca3)+({\LL*cos( 30)},{\LL*sin( 30)})$);
  \pic[rotate=120] at ($(Ca3)+({\PP*cos( 30)},{\PP*sin( 30)})$) {shgndD};
  \draw[edge] (Ca4) -- ($(Ca4)+({\LL*cos(270)},{\LL*sin(270)})$);
  \pic[rotate=  0] at ($(Ca4)+({\PP*cos(270)},{\PP*sin(270)})$) {shgndD};

  \node[font=\scriptsize] at ($(Cr1) +({0.85*cos( 90)},{0.85*sin( 90)})$) {$r_1$};
  \node[font=\scriptsize] at ($(Cr2) +({0.85*cos(-30)},{0.85*sin(-30)})$) {$r_2$};
  \node[font=\scriptsize] at ($(Ca1) +({0.85*cos(210)},{0.85*sin(210)})$) {$a_1$};
  \node[font=\scriptsize] at ($(Ca2) +({0.95*cos(150)},{0.95*sin(150)})$) {$a_2$};
  \node[font=\scriptsize] at ($(Ca3) +({0.95*cos( 30)},{0.95*sin( 30)})$) {$a_3$};
  \node[font=\scriptsize] at ($(Ca4) +(0,-0.80)$) {$a_4$};

  \node[font=\scriptsize, anchor=south]
    (kqlbl) at (0.55, 1.78) {central graph};
  \draw[->, thick, black!65, shorten >=1pt]
    (kqlbl.south) -- (0.30, {\KR+0.20});
\end{scope}

\begin{scope}[shift={(-0.6,-3.05)}, font=\scriptsize]
  \node[aport]  (lA)  at (0,0) {};
  \node[right=1.2mm of lA, inner sep=1pt] (lAt) {antenna port};

  \node[rport]  (lR)  at ($(lAt.east)+(0.30,0)$) {};
  \node[right=1.2mm of lR, inner sep=1pt] (lRt) {RF-chain port};

  \node[susc]   (lS)  at ($(lRt.east)+(0.30,0)$) {};
  \node[right=1.2mm of lS, inner sep=1pt] (lSt) {tunable susceptance};

  \node[shunt]  (lH)  at ($(lSt.east)+(0.30,0)$) {};
  \node[right=1.2mm of lH, inner sep=1pt] {tunable shunt};
\end{scope}

\end{tikzpicture}}
\caption{Reciprocal lossless \ac{MiLAC} topologies for $N_R{=}4, L_R{=}2$: (a) fully-connected and (b) stem-connected.}
\label{fig:topologies}
\end{figure}
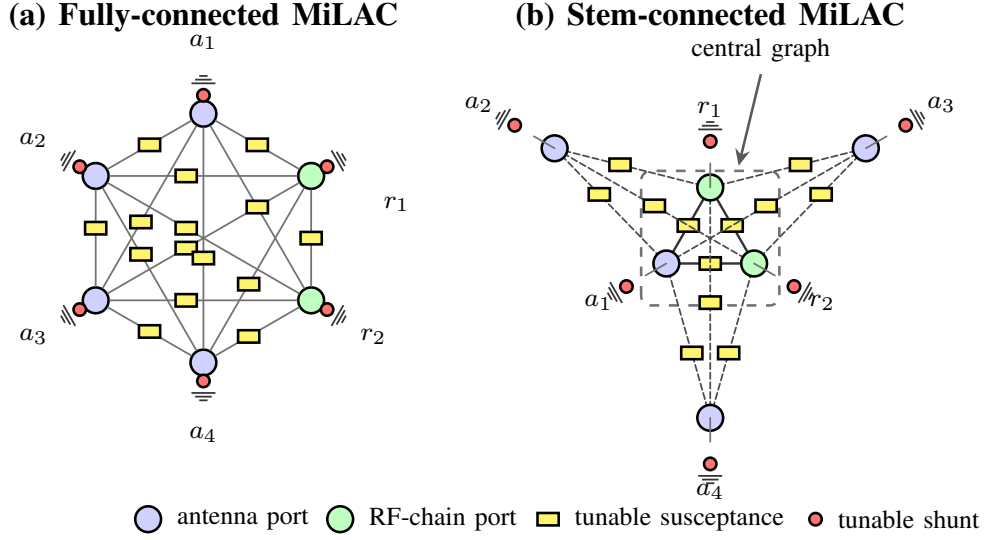

\section{Reduced-Complexity Hardware Design}
\label{sec:hardware}

Section~\ref{sec:crb-general} showed that $L_R\ge 2K$ \ac{RF} chains aligned to $\mathcal{S}_{K}^{\star}(\vect{\theta})$ match the fully-digital \ac{CRB} and that the corresponding combiner is realizable by some lossless reciprocal multiport, but left open how many tunable reactive components such a multiport actually requires. A dimension-counting argument lower-bounds this count linearly in $N_R$ for any Stiefel-universal \ac{MiLAC} class, and the stem-connected topology of~\cite{Nerini2025Reduced} attains the bound in the leading order, a sharp reduction from the $\mathcal{O}(N_R^{2})$ cost of a naive fully-connected realization. Section~\ref{subsec:lower-bound} relates this Stiefel-universality requirement to the point-to-point \ac{MIMO}-capacity argument of~\cite{Nerini2025Reduced}.

\subsection{Hardware cost of the fully-connected solution}
\label{subsec:fully-conn-cost}

A fully-connected reciprocal lossless $(N_R{+}L_R)$-port \ac{MiLAC} is a multiport microwave network (see Fig.~\ref{fig:topologies}(a)) in which every pair of ports is interconnected by a tunable reactive (purely imaginary) admittance, and every port additionally carries a tunable shunt-to-ground~\cite{Nerini2025Reduced}. Since the network is reciprocal, its admittance matrix $\Ym=\mathrm{j}\Bm$ is symmetric, where the susceptance matrix $\Bm\in\R^{M\times M}$ with $M:=N_R+L_R$ is real and symmetric. The tunable degrees of freedom are the $M$ diagonal entries (shunt susceptances) and the $\binom{M}{2}$ off-diagonal entries in the upper triangle, giving 
\begin{equation}
\label{eq:fully-conn-cost}
c_{\mathrm{fc}}(N_R,L_R)
=\tbinom{M+1}{2}=\mathcal{O}(N_R^{2})
\end{equation}
tunable susceptances. For a moderate array with $N_R=128$ and $L_R=2K$, evaluating~\eqref{eq:fully-conn-cost} gives $\binom{131}{2}=8515$ susceptances at $K=1$ and $\binom{137}{2}=9316$ at $K=4$. This quadratic scaling is the dominant hardware burden in large-aperture \ac{MiLAC} deployment and would offset much of the analog-hardware saving promised by the two-\ac{RF}-chain-per-target zero-gap theorem. Figure~\ref{fig:topologies} visualizes the topological difference underlying the $\mathcal{O}(N_R^2)\to\mathcal{O}(L_R N_R)$ reduction we set up.

\begin{remark}[Manifold interpretation of the fully-connected count]
\label{rem:manifold-count}
The count $M(M+1)/2$ is not accidental: it equals the real dimension of the manifold on which the scattering matrix $\Thetam$ lives, namely the set of $M\times M$ symmetric unitary matrices. The admittance-to-scattering map $\Thetam=(Y_0\Id-\mathrm{j}\Bm)(Y_0\Id+\mathrm{j}\Bm)^{-1}$~\cite{Nerini2025Reduced} is the Cayley transform, a smooth bijection between the real symmetric matrix $\Bm\in\R^{M\times M}$ and almost all such symmetric unitary matrices\footnote{The exceptional point corresponds to scattering matrices whose admittance representation does not exist. An alternative derivation via the Autonne--Takagi factorization~\cite{HornJohnson2013} gives the same count as $\dim_{\R}U(M)-\dim_{\R}O(M)=M^{2}-M(M-1)/2=M(M+1)/2$, where $U(M)$ and $O(M)$ are the $M\times M$ complex unitary and real orthogonal groups, respectively.}. Since a real symmetric $M\times M$ matrix has $M(M+1)/2$ independent entries, this manifold inherits the same real dimension, and the fully-connected topology has exactly the right number of tunable susceptances to sweep it. Whether the entire scattering manifold actually needs to be reached for \ac{CRB}-optimal sensing, and at what hardware floor, is the question taken up in Section~\ref{subsec:lower-bound}.
\end{remark}

\begin{table}[t]
\caption{Hardware complexity, architectural universality, and \ac{DoA} \ac{CRB} of receive architectures.}
\label{tab:complexity}
\centering
\footnotesize
\renewcommand{\arraystretch}{1.25}
\resizebox{\linewidth}{!}{%
\begin{tabular}{@{}l c c c c c@{}}
\toprule
Architecture & \ac{RF} chains $L_R$ & Tunable components & Scaling in $N_R$ & Stiefel-universal & \ac{DoA} \ac{CRB} \\
\midrule
Fully-digital baseline & $N_R$ & none & --- & --- & $\crb_{\dig}(\vect{\theta})$ \\
Fully-connected \ac{MiLAC} & $2K$ & $c_{\mathrm{fc}}=(N_R+L_R)(N_R+L_R+1)/2$ & $\mathcal{O}(N_R^{2})$ & \checkmark & $=\crb_{\dig}(\vect{\theta})$ \\
Stem-connected \ac{MiLAC} & $2K$ & $c_{\mathrm{stem}}=L_R(2N_R+1)$ & $\mathcal{O}(L_RN_R)$ & \checkmark & $\mathbf{{=}\crb_{\dig}(\vect{\theta})}$ \\
Stiefel-universal lower bound (Prop.~\ref{prop:lower-bound}) & $2K$ & $c_{\star}=2L_RN_R-L_R^{2}$ & $\mathcal{O}(L_RN_R)$ & \checkmark & $=\crb_{\dig}(\vect{\theta})$ \\
\bottomrule
\end{tabular}%
}
\end{table}

\subsection{A dimension-counting lower bound}
\label{subsec:lower-bound}

It is natural to ask whether the quadratic scaling is intrinsic to a digital-\ac{CRB}-preserving \ac{MiLAC}, or whether it can be reduced without sacrificing the zero \ac{CRB} gap. This parallels the capacity-literature trajectory in which the fully-connected \ac{MiLAC} of~\cite{Nerini2025Capacity} achieves the point-to-point \ac{MIMO} channel capacity with $\mathcal{O}(N_R^2)$ components and~\cite{Nerini2025Reduced} later showed that the stem-connected topology suffices with only $\mathcal{O}(L_RN_R)$ components. We mirror that trajectory: first ask what the minimum complexity of a \ac{CRB}-preserving \ac{MiLAC} is on purely dimension-theoretic grounds, then verify that the stem-connected topology attains it in the leading order. 

The link between the two settings is tighter than a loose analogy. To achieve the point-to-point \ac{MIMO} channel capacity for an \emph{arbitrary} channel, the receive \ac{MiLAC} of~\cite{Nerini2025Reduced} must realize that channel's capacity-optimal combiner, the projection onto its dominant left-singular subspace, and hence must reach \emph{every} orthonormal-column combiner. To attain the digital \ac{CRB} for an \emph{arbitrary} target configuration, the \ac{MiLAC} here must place its row space on the steering--derivative subspace $\mathcal{S}_{K}^{\star}(\vect{\theta})$, and hence must reach \emph{every} $2K$-dimensional subspace. By Proposition~\ref{prop:subspace}(ii) these two demands coincide: each holds precisely when the \ac{MiLAC} can realize an arbitrary orthonormal map, that is, when its reach covers the entire Stiefel manifold (Stiefel-universality, Definition~\ref{def:milac-class}). The two arguments therefore impose one and the same hardware requirement and differ only in the subspace the combiner must reach, the channel's eigen-subspace for communication versus the targets' steering--derivative subspace for sensing.

\begin{definition}[\ac{MiLAC} class, reach, Stiefel-universality]
\label{def:milac-class}
A \ac{MiLAC} class $\mathcal{C}$ is a triple $(\mathcal{T},\Phi,c)$, where $\mathcal{T}\subseteq\R^{c}$ is the parameter space of a fixed reciprocal lossless $(N_R{+}L_R)$-port interconnection topology with $c\ge 0$ tunable susceptances, and $\Phi:\mathcal{T}\to\C^{(N_R+L_R)\times(N_R+L_R)}$ is the real-analytic map assigning each susceptance vector $\mathbf{b}\in\mathcal{T}$ to the resulting symmetric unitary scattering matrix. The \emph{reach} of $\mathcal{C}$ is
\begin{equation}
\label{eq:reach-def}
\mathcal{G}(\mathcal{C}):=\{[\Phi(\mathbf{b})]_{21}:\mathbf{b}\in\mathcal{T}\}\subseteq\C^{L_R\times N_R}.
\end{equation}
We say that $\mathcal{C}$ is (i)~\emph{\ac{CRB}-preserving} if, for every $\vect{\theta}$, there exists $\Gm\in\mathcal{G}(\mathcal{C})$ with $\row(\Gm)\supseteq\mathcal{S}_{K}^{\star}(\vect{\theta})$, and (ii)~\emph{Stiefel-universal} if $\mathcal{G}(\mathcal{C})\supseteq\mathrm{St}(L_R,N_R)$.
\end{definition}

Stiefel-universality implies \ac{CRB}-preservation for every $\vect{\theta}$: any $L_R$-dimensional subspace of $\C^{N_R}$ containing $\mathcal{S}_{K}^{\star}(\vect{\theta})$ admits a row-isometric representative in $\mathrm{St}(L_R,N_R)$ by Proposition~\ref{prop:subspace}(ii), and by Stiefel-universality lies in $\mathcal{G}(\mathcal{C})$. The converse is false: a strictly \ac{CRB}-preserving class needs to cover only a $K$-parameter family of subspaces in $\mathrm{Gr}(L_R,N_R)$. We nevertheless target Stiefel-universality, for three reasons. First, a Stiefel-universal \ac{MiLAC} is deployable without prior knowledge of the \ac{DoA} distribution and is retunable to any target configuration at run time, whereas a strictly \ac{CRB}-preserving class would commit the hardware to a specific parameterized family of angles. Second, practical tracking requires the \ac{MiLAC} to realize $\Gm^{\star}(\hat{\vect{\theta}})$ over an open neighborhood of the nominal curve as $\hat{\vect{\theta}}$ is refined. Third, a Stiefel-universal topology handles arbitrary $K$ uniformly. We therefore study the minimum of $c(\mathcal{C})$ over Stiefel-universal classes.

\begin{proposition}[Lower bound on Stiefel-universal \ac{MiLAC} complexity]
\label{prop:lower-bound}
Let $\mathcal{C}=(\mathcal{T},\Phi,c)$ be a Stiefel-universal \ac{MiLAC} class. Then
\begin{equation}
\label{eq:lower-bound}
c\,\ge\,\dim_{\R}\mathrm{St}(L_R,N_R)\,=\,2L_RN_R-L_R^{2}.
\end{equation}
\end{proposition}

\begin{IEEEproof}
See Appendix~\ref{app:lower-bound}.
\end{IEEEproof}

For $L_R$ fixed and $N_R$ large, the right-hand side of~\eqref{eq:lower-bound} is $\mathcal{O}(L_RN_R)$, a factor of $\mathcal{O}(N_R/L_R)$ smaller than the fully-connected upper bound~\eqref{eq:fully-conn-cost}. Thus any \ac{CRB}-preserving \ac{MiLAC} \emph{could}, in principle, have linear rather than quadratic hardware cost. The question is whether an actual interconnection topology attains it.

\subsection{Stem-connected MiLACs attain the lower bound}
\label{subsec:stem-attains}

The stem-connected \ac{MiLAC} of~\cite{Nerini2025Reduced} is a reciprocal lossless multiport whose associated graph has $N_V:=N_R+L_R$ vertices (one per port) and is a \emph{center graph}~\cite{Nerini2025Reduced} with center size $Q:=2L_R-1$ (see Fig.~\ref{fig:topologies}(b)). The $Q$ central vertices are pairwise connected (a complete subgraph $K_Q$) and each is connected to every one of the $N_V-Q=N_R-L_R+1$ non-central vertices, which themselves carry no inter-port edges. For capacity-achievability on the receiver side, the central vertices must include all $L_R$ \ac{RF}-chain ports together with any $L_R-1$ antenna ports~\cite{Nerini2025Reduced}. Every port additionally carries a tunable shunt-to-ground. The edge count is
\begin{equation}
\label{eq:edge-count}
N_E=\frac{Q(Q-1)}{2}+Q(N_V-Q),
\end{equation}
yielding the total tunable-component count
\begin{equation}
\label{eq:stem-cost}
c_{\mathrm{stem}}(N_R,L_R)=N_V+N_E=L_R(2N_R+1),
\end{equation}
and~\cite{Nerini2025Reduced} proves that as these susceptances range over their physically admissible open sets, the $(2,1)$-block of the resulting scattering matrix sweeps out exactly $\mathrm{St}(L_R,N_R)$. Hence the stem-connected class is \emph{Stiefel-universal} in the sense of Definition~\ref{def:milac-class}. The inverse map, namely recovering susceptance values $\mathbf{b}$ such that $[\Phi(\mathbf{b})]_{21}=\Gm$ for a target $\Gm\in\mathrm{St}(L_R,N_R)$, follows in closed form from Algorithm~2 of~\cite{Nerini2025Reduced} in $\mathcal{O}(L_R^{2}N_R)$ arithmetic operations.

Combining~\eqref{eq:stem-cost} with~\eqref{eq:lower-bound} yields
\begin{equation}
\label{eq:stem-gap}
c_{\mathrm{stem}}(N_R,L_R)-(2L_RN_R-L_R^{2})\,=\,L_R(L_R+1),
\end{equation}
so the stem-connected architecture exceeds the dimension-counting lower bound by only $L_R(L_R+1)$ tunable components, an additive constant independent of $N_R$. In particular, stem-connected is \emph{order-optimal}: it matches the leading coefficient $2L_R$ of the lower bound exactly. For the cases of principal interest, namely $L_R=2K$ chains with up to $K=4$ targets, the overhead is at most $2K(2K+1)\le 72$, negligible relative to the leading $4KN_R$ term.

\begin{theorem}[Stem-connected \ac{MiLAC} is \ac{CRB}-optimal]
\label{thm:stem}
Assume $2K\le N_R$ and $\dim\mathcal{S}_{K}^{\star}(\vect{\theta})=2K$. For every $L_R\ge 2K$, the stem-connected \ac{MiLAC} of~\cite{Nerini2025Reduced}, with $c_{\mathrm{stem}}(N_R,L_R)=L_R(2N_R+1)$ tunable susceptances, realizes a feasible $\Gm\in\mathcal{F}_{\MiLAC}$ with $\row(\Gm)\supseteq\mathcal{S}_{K}^{\star}(\vect{\theta})$ and therefore attains
\begin{equation*}
\Jm_{\MiLAC}(\vect{\xi};\Gm)=\Jm_{\dig}(\vect{\xi}),\;\crb_{\MiLAC}(\vect{\theta};\Gm)=\crb_{\dig}(\vect{\theta}).
\end{equation*}
The stem-connected architecture closes the zero \ac{CRB} gap with $\mathcal{O}(L_RN_R)=\mathcal{O}(KN_R)$ tunable components, a factor of $\mathcal{O}(N_R/L_R)$ fewer than the fully-connected realization of~\eqref{eq:fully-conn-cost}, and within an additive constant $L_R(L_R+1)$ of the absolute lower bound~\eqref{eq:lower-bound}.
\end{theorem}

\begin{IEEEproof}
See Appendix~\ref{app:thm-stem}.
\end{IEEEproof}

Theorem~\ref{thm:stem} is the central hardware result of the paper. It closes the loop from information-theoretic \ac{CRB}-optimality (Theorem~\ref{thm:lowner}), through feasibility by a lossless reciprocal multiport (Lemma~\ref{lem:reachability}), all the way to a concrete physically-realized architecture with closed-form synthesis and linear-in-$(N_R,K)$ component count. 

\begin{remark}[Tightness of the lower bound]
\label{rem:tightness}
We do not claim that the $L_R(L_R+1)$ additive overhead of stem-connected over the lower bound~\eqref{eq:lower-bound} is unavoidable. The bound is a pure dimension-counting argument and does not exploit the physical reactive-component constraints. Whether a reciprocal lossless interconnection topology with exactly $2L_RN_R-L_R^{2}$ tunable susceptances that is Stiefel-universal exists is, to our knowledge, an open problem. The $\mathcal{O}(L_R^{2})$ overhead is of theoretical interest only and has no bearing on the asymptotic $\mathcal{O}(N_R/L_R)$ saving of stem-connected over the fully-connected synthesis.
\end{remark}

\subsection{Hardware complexity summary}
\label{subsec:complexity-summary}

Table~\ref{tab:complexity} summarizes the four architectures studied in this paper at the zero-gap threshold $L_R=2K$ of Corollary~\ref{cor:zerogap}, across five axes: number of \ac{RF} chains, exact tunable-component count, asymptotic scaling in $N_R$, Stiefel-universality, and achievable \ac{CRB} relative to the fully-digital benchmark.

To put the numbers in perspective, take $N_R=128$. At $K=1$ ($L_R=2$), the fully-connected \ac{MiLAC} needs $\binom{131}{2}=8515$ tunable susceptances, the Stiefel-universal lower bound evaluates to $508$, and the stem-connected \ac{MiLAC} needs $2(257)=514$, exceeding the lower bound by only $6$ components and yielding a $16.6\times$ reduction over the fully-connected count. At $K=4$ ($L_R=8$), the fully-connected count grows to $9316$, the lower bound is $1984$, and the stem-connected \ac{MiLAC} uses $8(257)=2056$ components, a $4.5\times$ reduction.

\section{Numerical Simulations}
\label{sec:numerical}

We numerically validate the theory of Section~\ref{sec:crb-general} on a half-wavelength \ac{ULA}. The experiments traverse per-target \ac{CRB} versus \ac{SNR} (Section~\ref{subsec:sim-snr}), aperture scaling versus $N_R$ (Section~\ref{subsec:sim-Nscaling}), the zero-gap threshold evaluation (Section~\ref{subsec:sim-LR}), per-target \ac{CRB} versus angular separation (Section~\ref{subsec:sim-geom}), and steering-mismatch robustness (Section~\ref{subsec:sim-mismatch}). Hardware complexity is not separately plotted: the closed-form counts of Section~\ref{sec:hardware} already capture the fully-connected versus stem-connected tradeoff.

\begin{table}[t]
\caption{Default simulation parameters.}
\label{tab:sim-defaults}
\centering
\footnotesize
\renewcommand{\arraystretch}{1.18}
\begin{tabular}{l c}
\toprule
Parameter & Value \\
\midrule
Array geometry & half-wavelength \ac{ULA} \\
Steering vector, $[\av(\theta)]_n$ & $e^{j\pi n\sin\theta}$, $n=0,\ldots,N_R{-}1$ \\
Number of snapshots, $T$ & $50$ (unit-energy symbols) \\
Pre-combining noise, $\nv_t$ & $\nv_t\sim\mathcal{CN}(\mathbf{0},\sigma^{2}\Id_{N_R})$, $\sigma^{2}=1$ \\
Number of antennas, $N_R$ & $32$  \\
Number of \ac{RF} chains, $L_R$ & $2K$ \\
Single-target angle, $\theta$ & $20^{\circ}$ \\
Two-target angles, $(\theta_1,\theta_2)$ & $(15^{\circ},25^{\circ})$ \\
SNR definition & $|\beta_k|^{2}/\sigma^{2}$ \\
Default SNR & $10$~dB \\
Target amplitude phases, $\arg\beta_k$ & $\mathcal{U}[0,2\pi)$ \\
\acs{MLE} Monte Carlo trials & $100$ \\
\bottomrule
\end{tabular}
\end{table}

\subsection{Default simulation settings and baselines}
\label{subsec:sim-defaults}

The default simulation parameters are collected in Table~\ref{tab:sim-defaults}. Per-figure deviations from these defaults are flagged in the corresponding captions only when they occur. Throughout, ``per-target \ac{CRB}'' is the diagonal entry of $\crb(\vect{\theta};\Gm)$ associated with the $k$-th target, reported in $\deg^{2}$. We compare three receive front ends throughout.
\begin{itemize}
\item \emph{Digital baseline:} $\Gm=\Id_{N_R}$ with $L_R=N_R$ \ac{RF} chains, attaining $\crb_{\dig}$.
\item \emph{Optimal \ac{MiLAC} $\Gm^{\star}$:} The row-isometric combiner whose rows span the dominant left-singular subspace of $\Wm(\vect{\theta})=[\av(\theta_1),\dot{\av}(\theta_1),\ldots,\av(\theta_K),\dot{\av}(\theta_K)]\in\C^{N_R\times 2K}$. Explicitly, taking the compact \ac{SVD} $\Wm=\Um\Sigmam\Vm^{\HH}$, we set $\Gm^{\star}=\Um_{[:,1:L_R]}^{\HH}$. For $L_R\ge 2K$ this realizes $\row(\Gm^{\star})\supseteq\mathcal{S}_{K}^{\star}(\vect{\theta})$ so the zero-gap condition of Corollary~\ref{cor:zerogap} is met, and for $L_R<2K$ the truncated \ac{SVD} retains the highest-energy directions of $\Wm$.
\item \emph{Phase-shifter combiner:} An $L_R\times N_R$ constant-modulus matrix $\Fm$ with $[\Fm]_{\ell,n}=\tfrac{1}{\sqrt{N_R}}e^{j\phi_{\ell,n}}$. To compare against the \ac{MiLAC} on a like-for-like basis, we let the phase-shifter combiner exploit the \emph{same} angle prior: $\Fm$ is the phase-projection of the optimal \ac{MiLAC} combiner, $\Fm=\tfrac{1}{\sqrt{N_R}}e^{j\angle\Gm^{\star}}$, i.e., the entrywise phase of $\Gm^{\star}$ renormalized to unit modulus. This is the natural unit-modulus counterpart of $\Gm^{\star}$ and keeps the row space as close as a constant-modulus matrix can to $\mathcal{S}_{K}^{\star}(\vect{\theta})$.
\end{itemize}

Since the \ac{CRB} is a lower bound, to confirm that it is attainable in practice, we also report the empirical \ac{MSE} of an actual estimator. Given a combiner $\Gm$, a shared known pilot $\{s_t\}_{t=1}^{T}$, and combined snapshots $\{\zv_t\}_{t=1}^{T}$, we use the exact concentrated $K$-target \ac{MLE}. The complex amplitudes $\vect{\beta}$ are profiled out in closed form, leaving the \ac{DoA} cost
\begin{equation}
\label{eq:mle-conc}
\ell(\vect{\theta})\,=\,\bar{\zv}^{\HH}\Rm_{\Gm}^{-1}\Mm(\vect{\theta})\bigl(\Mm(\vect{\theta})^{\HH}\Rm_{\Gm}^{-1}\Mm(\vect{\theta})\bigr)^{\!-1}\Mm(\vect{\theta})^{\HH}\Rm_{\Gm}^{-1}\bar{\zv},\notag
\end{equation}
where $\Mm(\vect{\theta})=\Gm\Am(\vect{\theta})$ and $\bar{\zv}=\sum_{t}s_{t}^{\ast}\zv_{t}$. The maximizer $\hat{\vect{\theta}}_{\mathrm{ML}}=\argmax_{\vect{\theta}}\ell(\vect{\theta})$ is obtained by a $K$-dimensional coordinate-wise grid search with parabolic-vertex polish. Under standard regularity $\hat{\vect{\theta}}_{\mathrm{ML}}$ is asymptotically efficient, so the Monte Carlo \ac{MSE} curve provides a direct empirical witness that the \ac{CRB} is tight for an implementable estimator. 

\subsection{Per-target CRB versus SNR}
\label{subsec:sim-snr}

Figure~\ref{fig:multi-snr} plots the average per-target \ac{CRB} versus \ac{SNR} for $K=2$ at the default $N_R=32$. Four observations confirm the theory. \emph{First,} the optimal \ac{MiLAC} at the zero-gap threshold $L_R=2K=4$ overlaps the digital curve, matching $\crb_{\MiLAC}^{\star}=\crb_{\dig}$ in Corollary~\ref{cor:zerogap}. \emph{Second,} reducing the front end by a single \ac{RF} chain to $L_R=3$ incurs a finite, \ac{SNR}-independent penalty of $\crb_{\MiLAC}^{\star}/\crb_{\dig}\approx 1.58$. Identifiability is preserved here because $L_R=3\ge\lceil 3K/2\rceil=3$ (Remark~\ref{rem:thresholds}), but the truncated row space loses one Fisher-information direction. \emph{Third,} the prior-informed phase-shifter combiner with the same $L_R=4$ sits about $1.7\times$ above the optimum: although it uses the same angle prior as the \ac{MiLAC}, the unit-modulus constraint keeps its row space from exactly containing $\mathcal{S}_{K}^{\star}(\vect{\theta})$, leaving a small but non-zero gap (Proposition~\ref{prop:ps-strict}). \emph{Fourth,} the empirical \ac{MLE} \ac{MSE} through both the optimal \ac{MiLAC} (red $\times$) and the digital baseline (purple $+$) sits on the $\crb_{\dig}=\crb_{\MiLAC}^{\star}$ curve at every \ac{SNR} (over $100$ Monte Carlo trials), confirming that the bound is attainable by an implementable estimator.

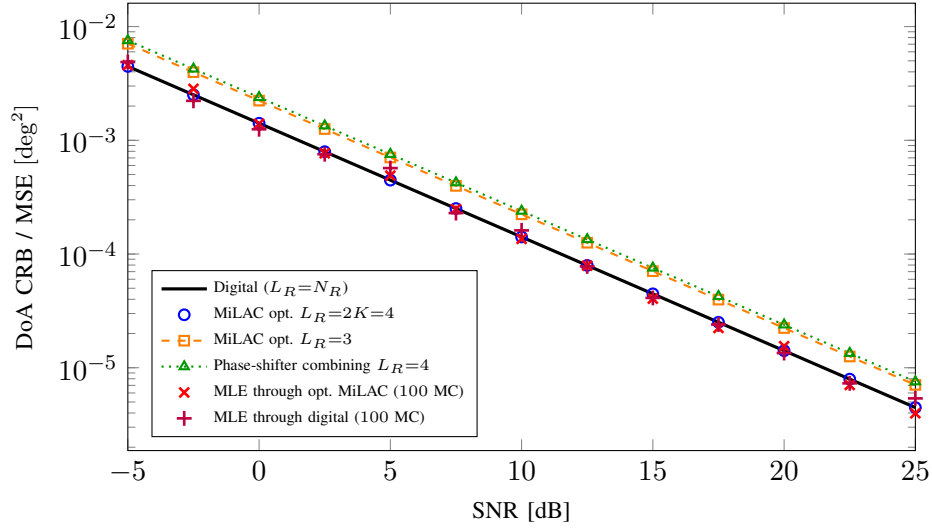
\begin{figure}[!tb]
\centering
\begin{tikzpicture}
\begin{semilogyaxis}[
    width=12cm, height=7.5cm,
    xlabel={SNR [dB]},
    ylabel={DoA CRB / MSE $[\deg^{2}]$},
    xmin=-5, xmax=25,
    legend pos=south west,
    legend cell align={left},
    legend style={font=\tiny, fill=white, fill opacity=0.85, draw opacity=1, text opacity=1},
]
\addplot[black, very thick] table[x=snr_db, y=dig_deg2] {figures/data/multi_crb_vs_snr.dat};
\addlegendentry{Digital ($L_R{=}N_R$)}
\addplot[blue, only marks, mark=o, mark size=2pt, mark options={line width=0.8pt}]
    table[x=snr_db, y=lr4_deg2] {figures/data/multi_crb_vs_snr.dat};
\addlegendentry{MiLAC opt.\ $L_R{=}2K{=}4$}
\addplot[orange, dashed, thick, mark=square, mark size=1.8pt, mark options={solid}]
    table[x=snr_db, y=lr3_deg2] {figures/data/multi_crb_vs_snr.dat};
\addlegendentry{MiLAC opt.\ $L_R{=}3$}
\addplot[green!60!black, dotted, thick, mark=triangle, mark size=2pt, mark options={solid}]
    table[x=snr_db, y=hyb4_deg2] {figures/data/multi_crb_vs_snr.dat};
\addlegendentry{Phase-shifter combining $L_R{=}4$}
\addplot[red, only marks, mark=x, mark size=2.6pt, mark options={line width=1pt}]
    table[x=snr_db, y=mle_opt_deg2] {figures/data/multi_crb_vs_snr.dat};
\addlegendentry{MLE through opt.\ MiLAC ($100$ MC)}
\addplot[purple, only marks, mark=+, mark size=2.8pt, mark options={line width=1pt}]
    table[x=snr_db, y=mle_dig_deg2] {figures/data/multi_crb_vs_snr.dat};
\addlegendentry{MLE through digital ($100$ MC)}
\end{semilogyaxis}
\end{tikzpicture}
\caption{Average per-target \ac{CRB} (lines) and \ac{MLE} Monte Carlo \ac{MSE} (markers) versus \ac{SNR} when $K=2$.}
\label{fig:multi-snr}
\end{figure}

\subsection{Antenna count scaling: CRB versus $N_R$}
\label{subsec:sim-Nscaling}

Figure~\ref{fig:multi-N} plots the per-target \ac{CRB} versus $N_R$ for $K=2$ targets. The digital and optimal-\ac{MiLAC} curves coincide and lie on the $\mathcal{O}(N_R^{-3})$ reference slope of Proposition~\ref{prop:ula-scaling}, confirming that allocating $L_R=2K$ \ac{RF} chains preserves the cubic aperture gain. The deviation below $N_R=16$ is a resolution effect: there the beamwidth $\approx 60^\circ/N_R$ is comparable to the $10^\circ$ separation, so the marginally resolved targets inflate the \ac{CRB} until the aperture resolves them. The prior-informed phase-shifter combiner tracks the same slope at a nearly constant $1.5$--$1.9\times$ ($2$--$3$~dB) offset, the residual price of the unit-modulus constraint that the \ac{MiLAC} removes by realizing $\mathcal{S}_{K}^{\star}(\vect{\theta})$ exactly. The two overlaid \ac{MLE} \ac{MSE} markers (red $\times$ through the optimal \ac{MiLAC}, purple $+$ through the digital baseline) sit on top of their respective \ac{CRB} curves at every $N_R$, confirming again that the $N_R^{-3}$ aperture gain is realized by an implementable estimator and not only by the bound.

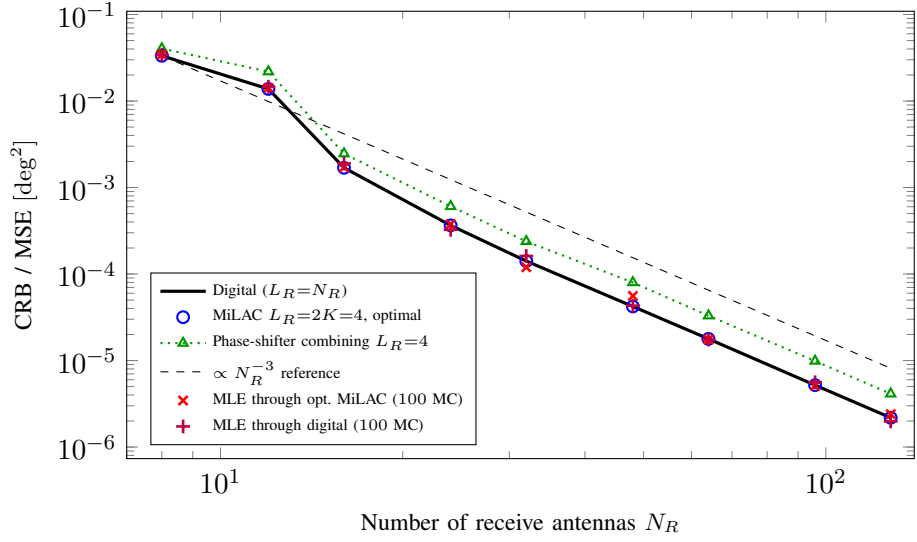
\begin{figure}[!tb]
\centering
\begin{tikzpicture}
\begin{loglogaxis}[
    width=12cm, height=7.5cm,
    xlabel={Number of receive antennas $N_R$},
    ylabel={CRB / MSE $[\deg^{2}]$},
    xmin=7, xmax=140,
    legend pos=south west,
    legend cell align={left},
    legend style={font=\tiny, fill=white, fill opacity=0.85, draw opacity=1, text opacity=1},
]
\addplot[black, very thick] table[x=N_R, y=dig_deg2] {figures/data/multi_crb_vs_N.dat};
\addlegendentry{Digital ($L_R{=}N_R$)}
\addplot[blue, only marks, mark=o, mark size=2.2pt, mark options={line width=0.8pt}]
    table[x=N_R, y=opt_deg2] {figures/data/multi_crb_vs_N.dat};
\addlegendentry{MiLAC $L_R{=}2K{=}4$, optimal}
\addplot[green!60!black, dotted, thick, mark=triangle, mark size=2pt, mark options={solid}]
    table[x=N_R, y=hyb_deg2] {figures/data/multi_crb_vs_N.dat};
\addlegendentry{Phase-shifter combining $L_R{=}4$}
\addplot[black, dashed] table[x=N_R, y=ref_deg2] {figures/data/multi_crb_vs_N.dat};
\addlegendentry{$\propto N_R^{-3}$ reference}
\addplot[red, only marks, mark=x, mark size=2.4pt, mark options={line width=1pt}]
    table[x=N_R, y=mle_opt_deg2] {figures/data/multi_crb_vs_N.dat};
\addlegendentry{MLE through opt.\ MiLAC ($100$ MC)}
\addplot[purple, only marks, mark=+, mark size=2.6pt, mark options={line width=1pt}]
    table[x=N_R, y=mle_dig_deg2] {figures/data/multi_crb_vs_N.dat};
\addlegendentry{MLE through digital ($100$ MC)}
\end{loglogaxis}
\end{tikzpicture}
\caption{Per-target \ac{CRB} versus $N_R$ when $K=2$.}
\label{fig:multi-N}
\end{figure}

\subsection{Zero-gap threshold evaluation}
\label{subsec:sim-LR}

The three-regime structure predicted by Remark~\ref{rem:thresholds} first becomes visually unambiguous at $K=3$, where the identifiability threshold $\lceil 3K/2\rceil=5$ and the zero-gap threshold $2K=6$ leave the unidentifiable regime ($L_R\le 4$), the intermediate-suboptimal regime ($L_R=5$), and the zero-gap regime ($L_R\ge 6$). Figure~\ref{fig:multi-K3} plots the \ac{CRB} ratio at $K=3$. For $L_R\le 4<5$ the $3K=9$-dimensional parameter vector is unidentifiable from $2L_R\le 8$ real measurements per snapshot and the \ac{CRB} is infinite (shown as red crosses at the top of the panel). At $L_R=5$ the problem becomes identifiable, and the optimal \ac{MiLAC} attains a finite, strictly positive gap of $\crb_{\MiLAC}^{\star}/\crb_{\dig}\approx 1.6$, instantiating the intermediate-regime penalty of Remark~\ref{rem:thresholds}. At $L_R=6=2K$ the gap closes exactly to unity and remains there for $L_R\ge 6$, matching Corollary~\ref{cor:zerogap}.

\begin{figure}[!tb]
\centering
\begin{tikzpicture}
\begin{semilogyaxis}[
    width=12cm, height=7.5cm,
    xlabel={Number of RF chains $L_R$},
    ylabel={$\mathrm{CRB}^{\star}_{\mathrm{MiLAC}}/\mathrm{CRB}_{\mathrm{dig}}$},
    xmin=0.5, xmax=10.5,
    xtick={1,2,3,4,5,6,7,8,9,10},
    ymin=0.6, ymax=1e3,
    legend cell align={left},
    legend style={at={(0.98,0.62)}, anchor=east, font=\tiny,
                  fill=white, fill opacity=0.85, draw opacity=1, text opacity=1},
]
\fill[red, opacity=0.10] (axis cs:0.5,0.6) rectangle (axis cs:4.5,1e3);
\fill[orange, opacity=0.18] (axis cs:4.5,0.6) rectangle (axis cs:5.5,1e3);
\fill[green!60!black, opacity=0.10] (axis cs:5.5,0.6) rectangle (axis cs:10.5,1e3);

\node[red!50!black, font=\tiny\bfseries] at (axis cs:2.5,400) {non-identifiable};
\node[orange!60!black, font=\tiny\bfseries, align=center] at (axis cs:5,400) {ident.\\+ gap};
\node[green!40!black, font=\tiny\bfseries] at (axis cs:8,400) {zero gap};

\addplot[red, only marks, mark=x, mark size=3pt, mark options={line width=1.2pt}]
    coordinates {(1,80) (2,80) (3,80) (4,80)};
\addlegendentry{CRB $=\infty$ (non-identifiable)}

\draw[red, dotted, thick] (axis cs:5,0.6) -- (axis cs:5,1e3);
\draw[green!60!black, dotted, thick] (axis cs:6,0.6) -- (axis cs:6,1e3);
\draw[black, dashed] (axis cs:0.5,1) -- (axis cs:10.5,1);

\addplot[blue, thick, mark=o, mark size=2.5pt, mark options={fill=white, line width=0.9pt}]
    table[x=LR, y=ratio] {figures/data/multi_identifiability_K3.dat};
\addlegendentry{MiLAC optimal CRB ratio}
\end{semilogyaxis}
\end{tikzpicture}
\caption{Per-target \ac{CRB} ratio versus $L_R$ at $K=3$, where red crosses mark the non-identifiable regime $L_R\le 4$.}
\label{fig:multi-K3}
\end{figure}
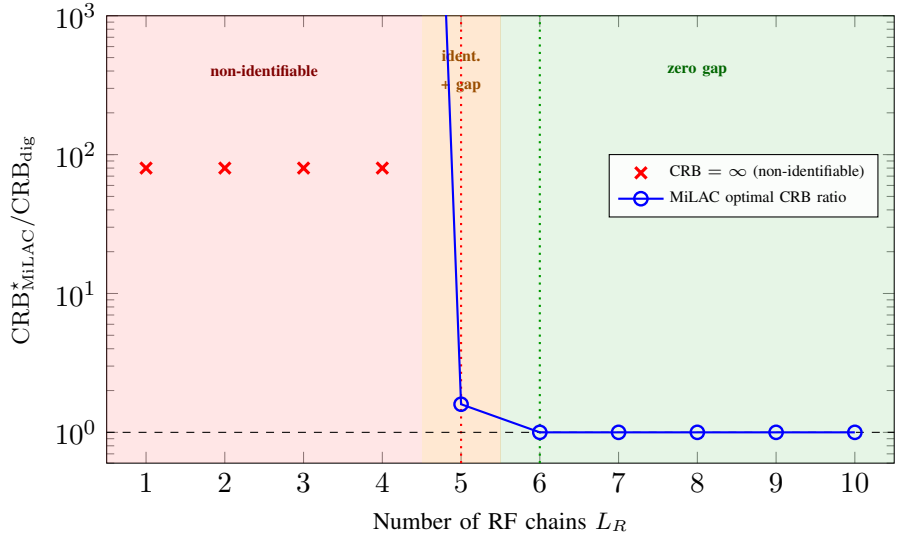

This three-regime structure is not special to $K=3$. Figure~\ref{fig:multi-K} sweeps the chain budget for $K\in\{1,2,3,4\}$ at $N_R=32$ and plots the same ratio against the offset $L_R-2K$. For every $K$, the ratio equals one throughout the zero-gap regime $L_R\ge 2K$ (offset $\ge 0$) and rises as soon as a chain is removed, so the onset of zero gap sits at offset zero independently of $K$. The curves part company only below the threshold: at offset $-1$ the problem stays identifiable for $K\ge 2$ with a finite intermediate-regime penalty (about $2.1$, $1.6$, and $1.5$ for $K=2,3,4$, the $K=3$ value reproducing the $L_R=5$ entry of Fig.~\ref{fig:multi-K3}), while deeper offsets are unidentifiable. The zero-gap threshold $L_R\ge 2K$ therefore grows linearly with the target count and is otherwise independent of it, confirming Corollary~\ref{cor:zerogap}.

\begin{figure}[!tb]
\centering
\begin{tikzpicture}
\begin{semilogyaxis}[
    width=12cm, height=7.5cm,
    ymode=log,
    ymin=0.55, ymax=200,
    xlabel={Offset $L_R-2K$},
    ylabel={$\mathrm{CRB}^{\star}_{\mathrm{MiLAC}}/\mathrm{CRB}_{\mathrm{dig}}$},
    xmin=-2, xmax=3,
    xtick={-3,-2,-1,0,1,2,3},
    legend pos=north east,
    legend cell align={left},
    legend style={font=\tiny, fill=white, fill opacity=0.85, draw opacity=1, text opacity=1},
]
\draw[black, dashed] (axis cs:-3.3,1) -- (axis cs:3.3,1);
\draw[black, dashed] (axis cs:0,0.55) -- (axis cs:0,200);
\node[black, font=\tiny, anchor=south] at (axis cs:0,80) {$L_R{=}2K$};
\addplot[blue, thick, mark=o, mark size=2pt, mark options={fill=white, line width=0.8pt}]
    table[x=offset, y=ratio] {figures/data/multi_crb_vs_K_K1.dat};
\addlegendentry{$K{=}1$}
\addplot[orange, thick, mark=square, mark size=2pt, mark options={fill=white, line width=0.8pt}]
    table[x=offset, y=ratio] {figures/data/multi_crb_vs_K_K2.dat};
\addlegendentry{$K{=}2$}
\addplot[green!60!black, thick, mark=triangle, mark size=2pt, mark options={fill=white, line width=0.8pt}]
    table[x=offset, y=ratio] {figures/data/multi_crb_vs_K_K3.dat};
\addlegendentry{$K{=}3$}
\addplot[red, thick, mark=diamond, mark size=2pt, mark options={fill=white, line width=0.8pt}]
    table[x=offset, y=ratio] {figures/data/multi_crb_vs_K_K4.dat};
\addlegendentry{$K{=}4$}
\end{semilogyaxis}
\end{tikzpicture}
\caption{Per-target \ac{CRB} ratio versus the chain-budget offset $L_R-2K$, for $K\in\{1,2,3,4\}$ at $N_R=32$.}
\label{fig:multi-K}
\end{figure}
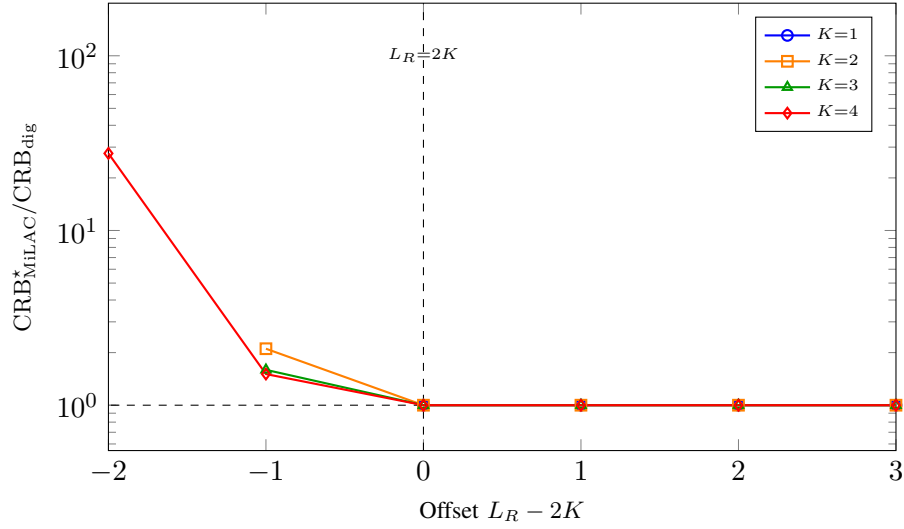

\subsection{Per-target CRB versus angular separation}
\label{subsec:sim-geom}

Figure~\ref{fig:multi-sep} shows the average per-target \ac{CRB} as a function of the angular separation $\Delta=\theta_2-\theta_1$, with $\theta_1=10^{\circ}$ fixed. We deviate from the default $N_R=32$ and set $N_R=16$ for this experiment so that the natural beamwidth $\theta_{\mathrm{b}}\approx 60^{\circ}/N_R\approx 3.8^{\circ}$ places the closely-spaced regime well inside the swept range $\Delta\in[2^{\circ},50^{\circ}]$. The $L_R=4=2K$ optimal \ac{MiLAC} tracks the digital \ac{CRB} at every separation, including the closely-spaced regime ($\Delta$ below $5^{\circ}$) where the \ac{CRB} rises steeply because the two steering vectors become highly correlated. With $L_R=3<2K$, the truncated three-dimensional row space cannot fully cover the four steering and derivative directions, so a finite gap to digital appears at every separation. The gap is moderate at small $\Delta$ and grows at intermediate and large $\Delta$ as the four columns of $[\av(\theta_1),\dot{\av}(\theta_1),\av(\theta_2),\dot{\av}(\theta_2)]$ become more linearly independent and the discarded singular direction carries an increasingly non-negligible share of Fisher information. The trace also oscillates for $\Delta$ above $30^{\circ}$, an artifact of \ac{SVD} truncation: at several values of $\Delta$ the third and fourth singular values cross, swapping which column the dominant three-dimensional left-singular subspace retains. The overlaid \ac{MLE} \ac{MSE} markers sit on the \ac{CRB} curves across the full range, including the steep $\Delta\to 0$ rise: the bound is sharp at the resolution boundary as well as at wide separations.

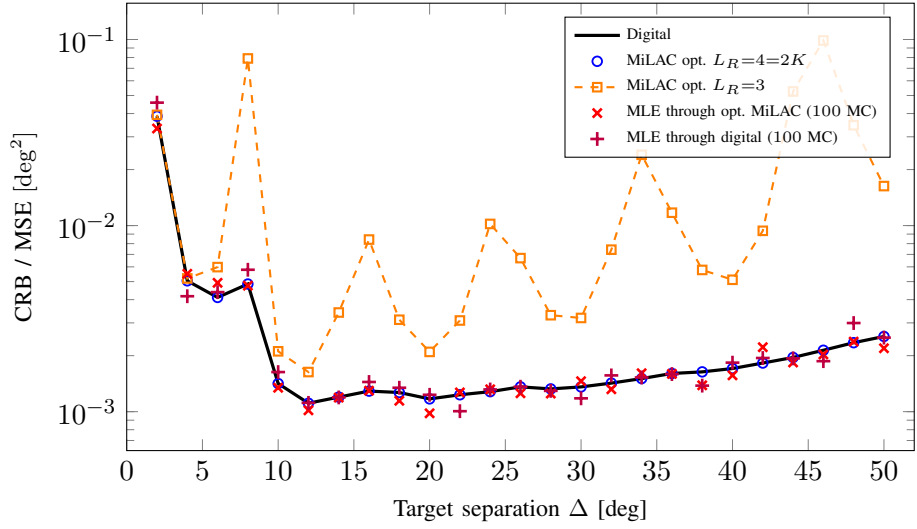
\begin{figure}[!tb]
\centering
\begin{tikzpicture}
\begin{semilogyaxis}[
    width=12cm, height=7.5cm,
    xlabel={Target separation $\Delta$ [deg]},
    ylabel={CRB / MSE $[\deg^{2}]$},
    xmin=0, xmax=52,
    legend pos=north east,
    legend cell align={left},
    legend style={font=\tiny, fill=white, fill opacity=0.85, draw opacity=1, text opacity=1},
]
\addplot[black, very thick] table[x=delta_deg, y expr={0.5*(\thisrow{dig_t1_deg2}+\thisrow{dig_t2_deg2})}]
    {figures/data/multi_crb_vs_sep.dat};
\addlegendentry{Digital}
\addplot[blue, only marks, mark=o, mark size=1.8pt, mark options={line width=0.7pt}]
    table[x=delta_deg, y expr={0.5*(\thisrow{lr4_t1_deg2}+\thisrow{lr4_t2_deg2})}]
    {figures/data/multi_crb_vs_sep.dat};
\addlegendentry{MiLAC opt.\ $L_R{=}4{=}2K$}
\addplot[orange, dashed, thick, mark=square, mark size=1.6pt, mark options={solid}]
    table[x=delta_deg, y expr={0.5*(\thisrow{lr3_t1_deg2}+\thisrow{lr3_t2_deg2})}]
    {figures/data/multi_crb_vs_sep.dat};
\addlegendentry{MiLAC opt.\ $L_R{=}3$}
\addplot[red, only marks, mark=x, mark size=2.4pt, mark options={line width=1pt}]
    table[x=delta_deg, y=mle_lr4_deg2] {figures/data/multi_crb_vs_sep.dat};
\addlegendentry{MLE through opt.\ MiLAC ($100$ MC)}
\addplot[purple, only marks, mark=+, mark size=2.6pt, mark options={line width=1pt}]
    table[x=delta_deg, y=mle_dig_deg2] {figures/data/multi_crb_vs_sep.dat};
\addlegendentry{MLE through digital ($100$ MC)}
\end{semilogyaxis}
\end{tikzpicture}
\caption{Average per-target \ac{CRB} versus angular separation $\Delta$ at $K=2$, with $\theta_1=10^{\circ}$, $\theta_2=10^{\circ}+\Delta$, $N_R=16$.}
\label{fig:multi-sep}
\end{figure}

\subsection{Robustness to angle mismatch}
\label{subsec:sim-mismatch}

Corollary~\ref{cor:zerogap} is an \emph{oracle} result: it assumes that the \ac{MiLAC} is steered to the \emph{true} target angle $\theta$, so that $\row(\Gm)\supseteq\mathcal{S}_{1}^{\star}(\theta)$. In practice the front end can only be steered to an estimate $\hat\theta$ produced by a coarse-acquisition stage (for instance a target-oblivious \ac{DFT} scan), and the question becomes: how quickly does the \ac{CRB} degrade when the combiner is slightly mis-oriented? Figure~\ref{fig:mismatch} reports the \ac{CRB} ratio
\begin{equation}
\label{eq:mismatch-ratio}
r(\Delta\hat\theta)\,:=\,\frac{\crb_{\MiLAC}\bigl(\theta;\Gm^{\star}(\theta{+}\Delta\hat\theta)\bigr)}{\crb_{\dig}(\theta)}
\end{equation}
as a function of $|\Delta\hat\theta|\in[0^{\circ},6^{\circ}]$, for $N_R\in\{16,32,64,128\}$, at $\theta=20^{\circ}$ and the single-target zero-gap budget $L_R=2K=2$. The ratio is $1$ at $\Delta\hat\theta=0$ (matching Corollary~\ref{cor:zerogap}) and grows as $|\Delta\hat\theta|$ increases, with a visible aperture dependence: the larger $N_R$, the sharper $\mathcal{S}_{1}^{\star}(\theta)$ on the Grassmannian and the more sensitive the \ac{CRB} to mis-steering.

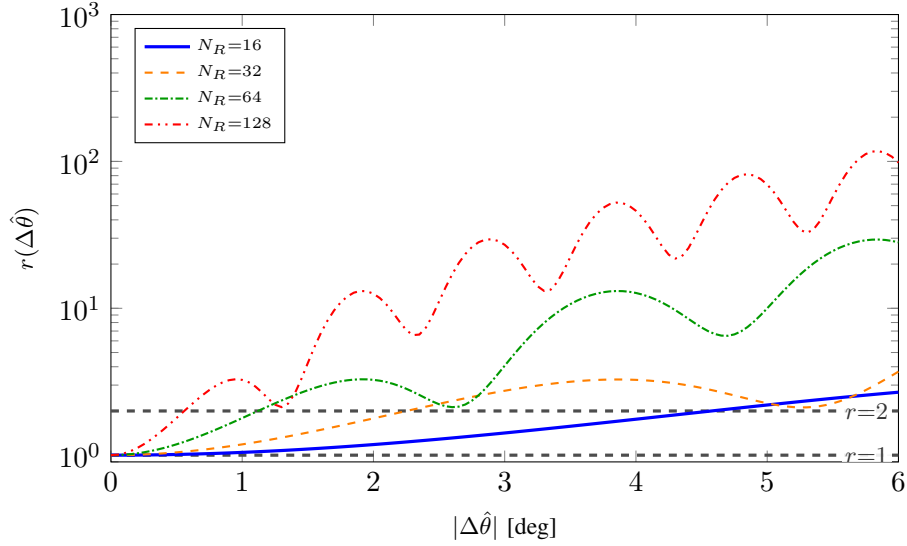
\begin{figure}[!tb]
\centering
\begin{tikzpicture}
\begin{semilogyaxis}[
    width=12cm, height=7.5cm,
    xlabel={$|\Delta\hat\theta|$ [deg]},
    ylabel={$r(\Delta\hat\theta)$},
    xmin=0, xmax=6,
    ymin=0.9, ymax=1e3,
    legend pos=north west,
    legend cell align={left},
    legend style={font=\tiny, fill=white, fill opacity=0.85, draw opacity=1, text opacity=1},
    clip=true,
]
\addplot[gray!60!black, very thick, dashed, forget plot, domain=0:6] {1};
\addplot[gray!60!black, very thick, dashed, forget plot, domain=0:6] {2};
\node[gray!40!black, font=\footnotesize, fill=white, fill opacity=0.9,
      text opacity=1, inner sep=1.2pt, anchor=east]
    at (axis cs:5.95,1)  {$r{=}1$};
\node[gray!40!black, font=\footnotesize, fill=white, fill opacity=0.9,
      text opacity=1, inner sep=1.2pt, anchor=east]
    at (axis cs:5.95,2)  {$r{=}2$};

\addplot[blue, very thick] table[x=delta_deg, y=r_N16] {figures/data/crb_vs_mismatch.dat};
\addlegendentry{$N_R{=}16$}
\addplot[orange, thick, dashed] table[x=delta_deg, y=r_N32] {figures/data/crb_vs_mismatch.dat};
\addlegendentry{$N_R{=}32$}
\addplot[green!60!black, thick, densely dashdotted] table[x=delta_deg, y=r_N64] {figures/data/crb_vs_mismatch.dat};
\addlegendentry{$N_R{=}64$}
\addplot[red, thick, dash dot dot] table[x=delta_deg, y=r_N128] {figures/data/crb_vs_mismatch.dat};
\addlegendentry{$N_R{=}128$}
\end{semilogyaxis}
\end{tikzpicture}
\caption{\ac{CRB} ratio $r(\Delta\hat\theta)$ versus steering mismatch $|\Delta\hat\theta|$ at $\theta=20^{\circ}$. The thick gray dashed lines mark the zero-gap reference $r=1$ and the $3$-dB penalty $r=2$.}
\label{fig:mismatch}
\end{figure}

To quote representative numbers, at a modest $N_R=32$ the penalty is only $4.5\%$ at $|\Delta\hat\theta|=0.5^{\circ}$, $19\%$ at $1^{\circ}$, and grows to $78\%$ at $2^{\circ}$. At $N_R=128$ the $3$-dB ($r=2$) penalty is reached near $|\Delta\hat\theta|\approx 0.55^{\circ}$. Beamwidth-scaled, the $3$-dB penalty consistently occurs near $|\Delta\hat\theta|\approx 0.4\theta_{\text{b}}$, where $\theta_{\text{b}}\propto 1/N_R$ is the natural beamwidth of the array. A coarse-acquisition stage of resolution $\ll\theta_{\text{b}}$ is therefore sufficient to realize the zero-gap promise of Corollary~\ref{cor:zerogap}, which a \ac{DFT} scan with $N_R$ beams naturally delivers.

\section{Conclusion}
\label{sec:conclusion}

This paper recasts \ac{MiLAC}-aided sensing as a Grassmannian geometry problem on $\mathrm{Gr}(L_R,N_R)$: the row space of the analog combiner is the sole information-bearing parameter, and the joint steering--derivative subspace $\mathcal{S}_{K}^{\star}(\vect{\theta})$ is the unique object it must span. Three insights summarize the message.

\emph{Two RF chains per target:} The zero-gap threshold $L_R\ge 2K$ is a hard floor on the chain count, separated from the identifiability threshold $\lceil 3K/2\rceil$ by an intermediate regime where the problem is identifiable but the \ac{MiLAC} \ac{CRB} sits strictly above $\crb_{\dig}$, so the active chain requirement is dictated by the target count rather than the array size.

\emph{Linear hardware cost:} The dimension-counting bound $2L_RN_R-L_R^{2}$ on any Stiefel-universal \ac{MiLAC} holds regardless of how the \ac{MiLAC}'s internal ports are wired together, i.e., independently of the circuit interconnection topology, and the stem-connected \ac{MiLAC} attains it asymptotically within an $N_R$-independent additive overhead. The resulting $\mathcal{O}(N_R/L_R)$ saving over the fully-connected \ac{MiLAC} is therefore generic to digital-\ac{CRB}-preserving architectures, rather than tied to one specific architecture, with a tunable-component count that stays linear in both the antenna and target counts.

\emph{Advantage over phase shifters:} A fixed constant-modulus phase-shifter combiner contains the target subspace $\mathcal{S}_{K}^{\star}(\vect{\theta})$, and hence matches $\crb_{\dig}$, only on a measure-zero set of configurations (Proposition~\ref{prop:ps-strict}), so it generally leaves a gap. The row span of such a combiner lives in a lower-dimensional subset of $\mathrm{Gr}(L_R,N_R)$, whereas a feasible \ac{MiLAC} hits $\mathcal{S}_{K}^{\star}(\vect{\theta})$ exactly with a perfectly conditioned orthonormal combiner.

A natural direction for future work is \ac{ISAC}: since the same Stiefel-universal \ac{MiLAC} attains both the communication and the sensing limits (Section~\ref{subsec:lower-bound}), a single shared front end is a promising candidate for serving both.

\appendices

\section{Proof of Theorem~\ref{thm:fim-projector}}
\label{app:fim-projector}

Stack the $T$ observations into $\zv=[\zv_1^{\TT},\ldots,\zv_T^{\TT}]^{\TT}\in\C^{TL_R}$, so that
\begin{equation}
\label{eq:stacked-obs}
\zv\,=\,(\Id_T\otimes\Gm)\,\vect{\mu}+\tilde{\nv},\qquad \tilde{\nv}\sim\mathcal{CN}\bigl(\mathbf{0},\sigma^{2}(\Id_T\otimes\Rm_{\Gm})\bigr),
\end{equation}
with $\vect{\mu}:=[\vect{\mu}_1^{\TT},\ldots,\vect{\mu}_T^{\TT}]^{\TT}\in\C^{TN_R}$. For a complex circular Gaussian observation with parameter-dependent mean $\vect{\mu}_{\mathrm{post}}=(\Id_T\otimes\Gm)\vect{\mu}$ and parameter-independent covariance $\Cm$, the Slepian--Bangs formula~\cite{Stoica2005} gives the \ac{FIM} entries
\begin{equation}
\label{eq:slepian-bangs}
[\Jm(\vect{\xi})]_{ij}\,=\,2\Re\Bigl\{\Bigl(\frac{\partial\vect{\mu}_{\mathrm{post}}}{\partial\xi_i}\Bigr)^{\HH}\Cm^{-1}\frac{\partial\vect{\mu}_{\mathrm{post}}}{\partial\xi_j}\Bigr\}.
\end{equation}
Collecting the partial derivatives into $\mathcal{J}_{\dig}$ and writing $\Cm=\sigma^{2}(\Id_T\otimes\Rm_{\Gm})$, this reads compactly as
\begin{equation}
\label{eq:SB-compact}
\Jm_{\MiLAC}(\vect{\xi};\Gm)\,=\,2\Re\Bigl\{\bigl((\Id_T\otimes\Gm)\mathcal{J}_{\dig}\bigr)^{\HH}\bigl(\Id_T\otimes\sigma^{-2}\Rm_{\Gm}^{-1}\bigr)\bigl((\Id_T\otimes\Gm)\mathcal{J}_{\dig}\bigr)\Bigr\}.
\end{equation}
The Kronecker identity
\begin{equation}
\label{eq:kron-id}
(\Am\otimes\Bm)^{\HH}(\Cm\otimes\Dm)(\Am\otimes\Bm)=(\Am^{\HH}\Cm\Am)\otimes(\Bm^{\HH}\Dm\Bm)
\end{equation}
simplifies the inner Kronecker factor of~\eqref{eq:SB-compact} to $\Id_T\otimes\sigma^{-2}\Gm^{\HH}\Rm_{\Gm}^{-1}\Gm$. Finally, the Moore--Penrose projector identity~\cite{HornJohnson2013}
\begin{equation}
\label{eq:mp-projector}
\Gm^{\HH}\Rm_{\Gm}^{-1}\Gm=\Pm_{\Gm}
\end{equation}
yields~\eqref{eq:SB-proj}. Setting $\Gm=\Id_{N_R}$, so that $\Pm_{\Gm}=\Id_{N_R}$, recovers~\eqref{eq:Jdig-def}.\hfill$\blacksquare$

\section{Proof of Proposition~\ref{prop:subspace}}
\label{app:subspace}

\emph{(i):} By Theorem~\ref{thm:fim-projector}, $\Jm_{\MiLAC}(\vect{\xi};\Gm)$ depends on $\Gm$ only through $\Pm_{\Gm}$. The projector $\Pm_{\Gm}$ is, in turn, the unique orthogonal projector onto $\row(\Gm)$, so it depends on $\Gm$ only through $\row(\Gm)$. Hence
\begin{equation}
\label{eq:rowspace-implies-fim}
\row(\Gm)=\row(\Gm')\ \Longrightarrow\ \Jm_{\MiLAC}(\vect{\xi};\Gm)=\Jm_{\MiLAC}(\vect{\xi};\Gm').
\end{equation}

\emph{(ii):} Given $\mathcal{S}\in\mathrm{Gr}(L_R,N_R)$, pick any full-row-rank $\Gm$ with $\row(\Gm)=\mathcal{S}$, and compute a reduced QR factorization $\Gm^{\HH}=\Qm\Rm$, with $\Qm^{\HH}\Qm=\Id_{L_R}$ and upper-triangular $\Rm\in\C^{L_R\times L_R}$. Setting $\Gm_{\mathrm{iso}}:=\Qm^{\HH}$ then gives
\begin{equation}
\label{eq:qr-iso}
\Gm_{\mathrm{iso}}\Gm_{\mathrm{iso}}^{\HH}=\Qm^{\HH}\Qm=\Id_{L_R},\qquad
\row(\Gm_{\mathrm{iso}})=\col(\Qm)=\col(\Gm^{\HH})=\row(\Gm)=\mathcal{S},
\end{equation}
so $\Gm_{\mathrm{iso}}$ is the desired row-isometry with row space $\mathcal{S}$.\hfill$\blacksquare$

\section{Proof of Theorem~\ref{thm:lowner}}
\label{app:lowner}

\emph{(i) Ordering:} Fix $\xv\in\R^{3K}$ and define $\yv(\xv):=\mathcal{J}_{\dig}\xv\in\C^{TN_R}$. From~\eqref{eq:SB-proj} and~\eqref{eq:Jdig-def},
\begin{equation}
\label{eq:quad-form}
\tfrac{\sigma^{2}}{2}\xv^{\TT}\Jm_{\MiLAC}\xv\,=\,\yv^{\HH}(\Id_T\otimes\Pm_{\Gm})\yv,\quad
\tfrac{\sigma^{2}}{2}\xv^{\TT}\Jm_{\dig}\xv\,=\,\yv^{\HH}\yv.
\end{equation}
Since $\Pm_{\Gm}\preceq\Id_{N_R}$ implies $\Id_T\otimes\Pm_{\Gm}\preceq\Id_{TN_R}$, the two identities in~\eqref{eq:quad-form} give
\begin{equation}
\label{eq:quad-ineq}
\xv^{\TT}\Jm_{\MiLAC}\xv\,=\,\tfrac{2}{\sigma^{2}}\yv^{\HH}(\Id_T\otimes\Pm_{\Gm})\yv\,\le\,\tfrac{2}{\sigma^{2}}\yv^{\HH}\yv\,=\,\xv^{\TT}\Jm_{\dig}\xv
\end{equation}
for every $\xv\in\R^{3K}$, which is~(i).

\emph{(ii) \ac{CRB} matrix ordering:} Partition the \ac{FIM} along the $\vect{\theta}/\vect{\beta}$ split into angle, amplitude, and cross blocks. The $\vect{\theta}$-marginal \ac{CRB}, that is, the angle bound once the unknown amplitudes have been optimally accounted for, equals the inverse of the Schur complement of the amplitude block,
\begin{equation}
\label{eq:theta-marginal-crb}
\bigl(\Jm_{\theta\theta}-\Jm_{\theta\beta}\Jm_{\beta\beta}^{-1}\Jm_{\beta\theta}\bigr)^{-1}.
\end{equation}
This Schur complement is monotone with respect to the L\"owner order~\cite{HornJohnson2013}: if one \ac{PSD} \ac{FIM} dominates another in the L\"owner sense, then so does its Schur complement. Combined with the fact that matrix inversion reverses the L\"owner order on \ac{PD} matrices, and with part~(i), this gives~(ii).

\emph{(iii) Equality:} The equality $\Jm_{\MiLAC}=\Jm_{\dig}$ holds if and only if $\xv^{\TT}(\Jm_{\dig}-\Jm_{\MiLAC})\xv=0$ for every $\xv\in\R^{3K}$, which by~\eqref{eq:quad-form} is in turn equivalent to
\begin{equation}
\label{eq:equality-projector}
(\Id_T\otimes\Pm_{\Gm})\,\yv(\xv)\,=\,\yv(\xv),\qquad\text{i.e.,}\qquad \yv(\xv)\in\C^{T}\otimes\mathcal{S}.
\end{equation}
Reading off the columns from~\eqref{eq:D-columns} and stacking, the score vector factorizes as
\begin{equation}
\label{eq:y-factor}
\begin{aligned}
\yv(\xv) &\,=\,\sv\otimes\vv(\xv),\\
\vv(\xv) &\,:=\,\sum_{k=1}^{K}\bigl[x_{\theta_k}\beta_k\dot{\av}(\theta_k)+(x_{\Re\{\beta_k\}}{+}jx_{\Im\{\beta_k\}})\av(\theta_k)\bigr],
\end{aligned}
\end{equation}
so that $\yv(\xv)\in\C^{T}\otimes\mathcal{S}$ if and only if $\vv(\xv)\in\mathcal{S}$. As $\xv$ ranges over $\R^{3K}$, the vectors $\vv(\xv)$ trace out the real-linear span
\begin{equation}
\label{eq:v-span}
\mathrm{span}_{\R}\bigl\{\av(\theta_k),\,j\av(\theta_k),\,\beta_k\dot{\av}(\theta_k)\bigr\}_{k=1}^{K}.
\end{equation}
Since $\mathcal{S}$ is a complex subspace, hence closed under multiplication by $-j$ and by $1/\beta_k\neq 0$, the containment $\vv(\xv)\in\mathcal{S}$ for all $\xv$ is equivalent to
\begin{equation}
\label{eq:steering-containment}
\bigl\{\av(\theta_k),\dot{\av}(\theta_k)\bigr\}_{k=1}^{K}\subseteq\mathcal{S},\qquad\text{i.e.,}\qquad \mathcal{S}_{K}^{\star}(\vect{\theta})\subseteq\mathcal{S}.
\end{equation}
\hfill$\blacksquare$

\section{Proof of Proposition~\ref{prop:ula-scaling}}
\label{app:ula-scaling}

We establish the aperture scaling in three steps: a Schur reduction of the angle information matrix, a Neumann-series control of the amplitude block, and the inversion of the resulting diagonally-dominant matrix. We first record the array bookkeeping the three steps rely on.

For the half-wavelength \ac{ULA}, define the exponential sums
\begin{equation}
\label{eq:exp-sums}
\sigma_{m}(\phi)\,:=\,\sum_{n=0}^{N_R-1}n^{m}e^{jn\phi}.
\end{equation}
At $\phi=0$ these reduce to
\begin{equation}
\label{eq:exp-sums-zero}
\sigma_0=N_R,\quad \sigma_1=\frac{N_R(N_R-1)}{2},\quad \sigma_2=\frac{N_R(N_R-1)(2N_R-1)}{6}.
\end{equation}
For any fixed $\phi\neq 0$, successive differentiation of the bounded geometric sum $\sigma_0(\phi)=(1-e^{jN_R\phi})/(1-e^{j\phi})$ yields $\sigma_m(\phi)=\mathcal{O}(N_R^{m})$. Writing $\phi_{kl}:=\pi(\sin\theta_l-\sin\theta_k)$, $c:=2\|\sv\|^{2}/\sigma^{2}$, and $\av_k:=\av(\theta_k)$, direct computation gives the \ac{ULA} inner products
\begin{equation}
\label{eq:ula-inner}
\av_k^{\HH}\av_l=\sigma_0(\phi_{kl}),\quad
\dot{\av}_k^{\HH}\av_l=-j\pi\cos\theta_k\,\sigma_1(\phi_{kl}),\quad
\dot{\av}_k^{\HH}\dot{\av}_l=\pi^{2}\cos\theta_k\cos\theta_l\,\sigma_2(\phi_{kl}).
\end{equation}

\emph{Step 1 (Schur reduction of the angle block):} Partition the $3K\times 3K$ digital \ac{FIM} into the $K\times K$ angle block $\Jm_{\theta\theta}$, the $2K\times 2K$ amplitude block $\Jm_{\beta\beta}$, and the cross-block $\Jm_{\theta\beta}$, and write the $\vect{\theta}$-marginal information matrix as the Schur complement
\begin{equation}
\label{eq:schur-ula}
\tilde{\Jm}_{\theta\theta}\,:=\,\Jm_{\theta\theta}-\Jm_{\theta\beta}\Jm_{\beta\beta}^{-1}\Jm_{\beta\theta}.
\end{equation}

\emph{Step 2 (amplitude block via a Neumann series):} Decompose the amplitude block as $\Jm_{\beta\beta}=\Jm_{\beta\beta}^{(\mathrm{d})}+\Jm_{\beta\beta}^{(\mathrm{o})}$ into its $K$ per-target $2\times 2$ diagonal blocks, each $\mathcal{O}(N_R)$, and the off-diagonal cross-target blocks, each $\mathcal{O}(1)$. Since $\|(\Jm_{\beta\beta}^{(\mathrm{d})})^{-1}\Jm_{\beta\beta}^{(\mathrm{o})}\|_\infty=\mathcal{O}(N_R^{-1})$, a Neumann series, the matrix analog of the geometric series $(1-x)^{-1}=1+x+x^{2}+\cdots$, gives
\begin{equation}
\label{eq:Jbb-inv}
\Jm_{\beta\beta}^{-1}\,=\,(\Jm_{\beta\beta}^{(\mathrm{d})})^{-1}+\mathcal{O}(N_R^{-2}).
\end{equation}
Substituting~\eqref{eq:Jbb-inv} into~\eqref{eq:schur-ula} and retaining the $k$-th target's own amplitude block yields the entries
\begin{equation}
\label{eq:Jtt-entries}
[\tilde{\Jm}_{\theta\theta}]_{kk}=c|\beta_k|^{2}\rho_{\dig}(\theta_k)+\mathcal{O}(N_R^{2}),\qquad
[\tilde{\Jm}_{\theta\theta}]_{kl}=\mathcal{O}(N_R^{2})\ (k\neq l),
\end{equation}
where $\rho_{\dig}(\theta_k)=\mathcal{O}(N_R^{3})$, and the off-diagonal estimate holds because every term in $[\tilde{\Jm}_{\theta\theta}]_{kl}$ involves a cross-target inner product.

\emph{Step 3 (inversion of a diagonally-dominant matrix):} Write $\tilde{\Jm}_{\theta\theta}=\Dm(\Id_K+\Em)$ with $\Dm=\diag([\tilde{\Jm}_{\theta\theta}]_{kk})=\mathcal{O}(N_R^3)$ and $\|\Em\|_{\infty}=\mathcal{O}(N_R^{-1})$. A Neumann-series expansion gives $(\Id_K+\Em)^{-1}=\Id_K+\mathcal{O}(N_R^{-1})$, hence
\begin{equation}
\label{eq:crb-diag-scaling}
[\crb_{\dig}(\vect\theta)]_{kk}=[\Dm^{-1}]_{kk}\bigl(1+\mathcal{O}(N_R^{-1})\bigr)=\mathcal{O}(N_R^{-3}).
\end{equation}
The \ac{MiLAC} claim follows from Corollary~\ref{cor:zerogap}.

For $K=1$ the cross-target terms are absent, and direct substitution of the closed-form values $\sigma_m(0)$ from~\eqref{eq:exp-sums-zero} yields
\begin{equation}
\label{eq:rho-dig-K1}
\rho_{\dig}(\theta)=\frac{\pi^{2}\cos^{2}\theta\,N_R(N_R^{2}-1)}{12},
\end{equation}
and hence~\eqref{eq:ula-K1}.\hfill$\blacksquare$

\section{Proof of Proposition~\ref{prop:lower-bound}}
\label{app:lower-bound}

\emph{Step 1 (Stiefel dimension):} Identify $\C^{L_R\times N_R}$ with $\R^{2L_R N_R}$ by separating real and imaginary parts, and recall that $\mathrm{St}(L_R,N_R)=\{\Gm\in\C^{L_R\times N_R}:\Gm\Gm^{\HH}=\Id_{L_R}\}$. The constraint $\Gm\Gm^{\HH}=\Id_{L_R}$ is an $L_R\times L_R$ Hermitian equation, so it amounts to $L_R$ real diagonal equations and $\binom{L_R}{2}$ complex strict-upper-triangle equations, i.e.,
\begin{equation}
\label{eq:stiefel-constraints}
L_R+L_R(L_R-1)=L_R^{2}
\end{equation}
independent real constraints. Since the differential (Jacobian) of the constraint map $\Gm\mapsto\Gm\Gm^{\HH}-\Id_{L_R}$ has full rank $L_R^{2}$ at every point of $\mathrm{St}(L_R,N_R)$~\cite{Edelman1998Stiefel}, the implicit function theorem makes $\mathrm{St}(L_R,N_R)$ a smooth submanifold (a smoothly curved surface) whose real dimension is
\begin{equation}
\label{eq:stiefel-dim}
d\,:=\,2L_RN_R-L_R^{2}.
\end{equation}

\emph{Step 2 (block projection):} Define $\Psi:\mathcal{T}\to\C^{L_R\times N_R}\cong\R^{2L_RN_R}$ by
\begin{equation}
\label{eq:psi-def}
\Psi(\mathbf{b})\,:=\,[\Phi(\mathbf{b})]_{21}.
\end{equation}
The admittance matrix of a reactive multiport is linear in its susceptances~\cite{Nerini2025AnalogI}, the scattering matrix is obtained from the admittance matrix via the map~\cite{Nerini2025Reduced}
\begin{equation}
\label{eq:cayley-app}
\Thetam=(Y_0\Id-\mathrm{j}\Bm)(Y_0\Id+\mathrm{j}\Bm)^{-1},
\end{equation}
and selecting the $(2,1)$ block is itself a linear operation. As the composition of the linear susceptance-to-admittance map, the rational map $\Bm\mapsto\Thetam$ in~\eqref{eq:cayley-app}, and this linear block selection, $\Psi$ is rational, hence smooth on $\mathcal{T}$.

\emph{Step 3 (dimension argument):} The idea is simple: a smooth map cannot raise dimension, so $c$ tunable susceptances cannot sweep out the $d$-dimensional manifold $\mathrm{St}(L_R,N_R)$ unless $c\ge d$. We make this precise with a volume argument. By Stiefel-universality, $\Psi(\mathcal{T})\supseteq\mathrm{St}(L_R,N_R)$. Suppose, for contradiction, that $c<d$. Near any $\Gm_{0}\in\mathrm{St}(L_R,N_R)$, the implicit function theorem supplies a smooth local coordinate chart $\phi:\mathcal{U}\to\R^{d}$, that is, a smooth and smoothly invertible map carrying a neighborhood of $\Gm_0$ onto an open subset of $\R^{d}$. Composing it with $\Psi$ gives a smooth map
\begin{equation}
\label{eq:composed-map}
\phi\circ\Psi:\ \Psi^{-1}(\mathcal{U})\subseteq\R^{c}\,\to\,\R^{d},\qquad c<d.
\end{equation}
Its image contains $\phi(\mathcal{U})$, a nonempty open subset of $\R^{d}$ of strictly positive $d$-dimensional volume (Lebesgue measure). However, a smooth map from a lower-dimensional Euclidean space ($c<d$) cannot cover any region of positive $d$-dimensional volume, as its image has measure zero, a standard consequence of Sard's theorem~\cite{Sard1942}. This contradicts the positive volume just exhibited. Hence
\begin{equation}
\label{eq:lower-bound-final}
c\,\ge\,d\,=\,2L_RN_R-L_R^{2},
\end{equation}
which establishes~\eqref{eq:lower-bound}.\hfill$\blacksquare$

\section{Proof of Theorem~\ref{thm:stem}}
\label{app:thm-stem}

The proof has two parts: first, a row-isometry attaining the digital \ac{CRB} exists, and second, the stem-connected topology realizes it.

\emph{Step 1 (an optimal combiner exists):} By Proposition~\ref{prop:subspace}(ii), for $L_R\ge 2K$ there exists a row-isometric $\Gm_{K}^{\star}\in\mathrm{St}(L_R,N_R)$ with $\row(\Gm_{K}^{\star})\supseteq\mathcal{S}_{K}^{\star}(\vect{\theta})$, and Lemma~\ref{lem:reachability} places it in $\mathcal{F}_{\MiLAC}$. Theorem~\ref{thm:lowner}(iii) then gives
\begin{equation}
\label{eq:stem-crb-equality}
\Jm_{\MiLAC}(\vect{\xi};\Gm_{K}^{\star})=\Jm_{\dig}(\vect{\xi}),
\end{equation}
and the corresponding \ac{CRB} equality.

\emph{Step 2 (realization by the stem-connected topology):} It remains to verify that the stem-connected topology realizes $\Gm_{K}^{\star}$. In~\cite{Nerini2025Reduced}, a receiver-side \ac{MiLAC} is called \emph{capacity-achieving} if, for every $\bar{\Um}\in\C^{N_R\times L_R}$ with orthonormal columns, there exists an admissible susceptance matrix such that
\begin{equation}
\label{eq:cap-achieving}
[\Thetam]_{21}=\bar{\Um}^{\HH}.
\end{equation}
As $\bar{\Um}$ ranges over all orthonormal-column matrices, $\bar{\Um}^{\HH}$ ranges over all of $\mathrm{St}(L_R,N_R)$, so ``capacity-achieving'' in~\cite{Nerini2025Reduced} is equivalent to Stiefel-universality in Definition~\ref{def:milac-class}. \cite{Nerini2025Reduced} proves that any center graph with center size $Q=2L_R-1$ whose central vertices include all $L_R$ \ac{RF}-chain ports is capacity-achieving, and the stem-connected topology is precisely such a center graph. Hence $\Gm_{K}^{\star}\in\mathcal{G}(\mathcal{C}_{\mathrm{stem}})$. Moreover, \cite{Nerini2025Reduced} produces a closed-form susceptance configuration $\mathbf{b}^{\star}$ with $[\Phi(\mathbf{b}^{\star})]_{21}=\Gm_{K}^{\star}$ in $\mathcal{O}(L_R^{2}N_R)$ arithmetic operations. The complexity comparison follows from~\eqref{eq:fully-conn-cost},~\eqref{eq:stem-cost}, and~\eqref{eq:lower-bound}.\hfill$\blacksquare$

\section{Proof of Proposition~\ref{prop:ps-strict}}
\label{app:ps-strict}

\emph{(i) L\"owner ordering:} The ordering and equality condition of Theorem~\ref{thm:lowner} depend only on the row-space projector $\Pm_{\row(\Fm)}$, so they apply verbatim to every full-row-rank $\Fm$ including any phase-shifter combiner $\Fm\in\mathcal{F}_{\mathrm{PS}}$, which is generically full row rank.

\emph{(ii) Genericity failure:} The map sending a phase-shifter combiner to its row space, $\row:\mathcal{F}_{\mathrm{PS}}\to\mathrm{Gr}(L_R,N_R)$, is real-analytic (locally given by a convergent power series). Left-multiplying any $\Fm\in\mathcal{F}_{\mathrm{PS}}$ by a diagonal phase matrix,
\begin{equation}
\label{eq:torus-action}
\Fm\,\mapsto\,\diag(e^{j\varphi_{1}},\ldots,e^{j\varphi_{L_R}})\,\Fm,
\end{equation}
leaves both its row space and the unit-modulus entry constraint unchanged. These row-wise phases form an $L_R$-parameter family, one phase $\varphi_{\ell}$ per row, so every reachable row space has at least an $L_R$-dimensional set of distinct preimages in $\mathcal{F}_{\mathrm{PS}}$. Consequently the reachable set $\mathcal{R}_{\mathrm{PS}}$ has real dimension at most
\begin{equation}
\label{eq:rps-dim}
\dim_{\R}\mathcal{R}_{\mathrm{PS}}\,\le\,\dim_{\R}\mathcal{F}_{\mathrm{PS}}-L_R\,=\,L_RN_R-L_R\,=\,L_R(N_R-1),
\end{equation}
whereas the target Grassmannian has real dimension
\begin{equation}
\label{eq:gr-dim}
\dim_{\R}\mathrm{Gr}(L_R,N_R)\,=\,2L_R(N_R-L_R).
\end{equation}
Subtracting, whenever $N_R\ge 2L_R$,
\begin{equation}
\label{eq:dim-gap}
\dim_{\R}\mathrm{Gr}(L_R,N_R)-\dim_{\R}\mathcal{R}_{\mathrm{PS}}\,\ge\,L_R(N_R-2L_R+1)\,=\,\Delta(L_R,N_R)\,>\,0.
\end{equation}
Since $\mathcal{R}_{\mathrm{PS}}$ is the image of a real-analytic map on a compact domain, it is contained in a finite union of smooth surfaces of dimension at most $L_R(N_R-1)<\dim_{\R}\mathrm{Gr}(L_R,N_R)$, and therefore has measure zero, occupying zero volume in $\mathrm{Gr}(L_R,N_R)$.

\emph{(iii):} Fix $\Fm\in\mathcal{F}_{\mathrm{PS}}$, so that $\row(\Fm)$ is a proper subspace of $\C^{N_R}$. The steering map $\theta\mapsto\av(\theta)$ is real-analytic, and \ac{ULA} steering vectors at distinct angles are linearly independent (Section~\ref{subsec:lowner}), so this curve is not contained in $\row(\Fm)$. Therefore the set $\{\theta:\av(\theta)\in\row(\Fm)\}$ is the zero set of a nontrivial real-analytic function, hence has measure zero, and a fortiori the set $\{\vect{\theta}:\mathcal{S}_{K}^{\star}(\vect{\theta})\subseteq\row(\Fm)\}$ has measure zero. By part~(i), it follows that
\begin{equation}
\label{eq:ps-strict-ineq}
\crb_{\mathrm{PS}}(\vect{\theta};\Fm)\,\succ\,\crb_{\dig}(\vect{\theta})
\end{equation}
for every $\vect{\theta}$ outside this measure-zero set. The \ac{MiLAC} claim follows from Lemma~\ref{lem:reachability} and Theorem~\ref{thm:lowner}(iii).\hfill$\blacksquare$

\bibliographystyle{IEEEtran}
\bibliography{references}

\end{document}